\documentclass[a4paper,cleveref, autoref, thm-restate]{lipics-v2021}

\pdfoutput=1
\usepackage{amsmath, amssymb, amsthm}
\usepackage{pifont}
\usepackage{bbold}
\usepackage{mathtools}
\usepackage{ebproof}
\ebproofset{center = false}
\usepackage{braket}
\usepackage{appendix}
\usepackage{stmaryrd}
\usepackage{todonotes}
\presetkeys{todonotes}{inline,backgroundcolor=orange}{}
\hypersetup{colorlinks,linkcolor={blue},citecolor={blue}}

\foreach \b in {C, N, R}{
  \expandafter\xdef\csname mb\b \endcsname{\noexpand\mathbb{\b}}
}
\foreach \c in {A, B, C, F, J, O, R, S, T, V, U, X}{
  \expandafter\xdef\csname mc\c \endcsname{\noexpand\mathcal{\c}}
}
\foreach \t in {a, b, c, f, g, h, n, q, x, y, Q}{
  \expandafter\xdef\csname mt\t \endcsname{\noexpand\mathtt{\t}}
}
\newcommand{\mfR}{\mathfrak{R}}
\newcommand{\fractwo}{\frac{1}{\sqrt 2}}
\newcommand{\gmid}{\,\mid\,}
\newcommand{\size}[1]{|#1|}
\newcommand{\computes}[1]{\llbracket #1 \rrbracket}
\newcommand{\fv}[1]{\mathtt{FV}(#1)}
\newcommand{\pitwo}{\Pi_2^0}
\newcommand{\kron}[2]{\delta_{#1, #2}}
\newcommand{\seq}[1]{\overrightarrow{#1}}
\newcommand{\fbqp}{\ensuremath{\mathtt{FBQP}}}

\foreach \n/\c in {no/red,struct/orange,det/teal}{
  \expandafter \xdef \csname thmprf\n \endcsname{\noexpand\color{\c}}
}

\newcommand{\arity}[1][\mtb]{\mathtt{ar}(#1)}
\newcommand{\var}[1]{\mathcal{V}ar(#1)}
\newcommand{\trsyntax}{\langle \mcA, \mcC, \mcF, \mcX, \mcR \rangle}
\newcommand{\trset}{\mathscr{T}}
\newcommand{\trsetground}{\mathscr{T}_0}
\newcommand{\trsetclassical}{\mathscr{C}}
\newcommand{\trsetval}{\mathscr{V}}
\newcommand{\trsamp}[1][\mfR]{\mcA}
\newcommand{\trsvar}[1][\mfR]{\mcX}
\newcommand{\trscons}[1][\mfR]{\mcC}
\newcommand{\trsfunc}[1][\mfR]{\mcF}
\newcommand{\trsrule}[1][\mfR]{\mcR_{#1}}
\newcommand{\symbols}{\mcC \biguplus \mcF}

\newcommand{\natz}{\mathtt{0}}
\newcommand{\nats}{\mathtt{S}}
\newcommand{\zket}{\ket 0}
\newcommand{\oket}{\ket 1}
\newcommand{\consp}{\mathtt{pr}}
\newcommand{\consht}{\mathtt{cons}}
\newcommand{\nil}{[~]}

\newcommand{\stwo}{\mta_{\fractwo}}

\DeclarePairedDelimiter{\inter}{\llbracket}{\rrbracket}

\newcommand{\can}{\mathtt{CAN}}
\newcommand{\toq}{\to_\mfR}
\newcommand{\toqif}{\toq^?}
\newcommand{\toqs}{\toq^*}
\newcommand{\nf}{\mathtt{NF}_\mfR}
\newcommand{\nequivz}[1][t]{#1 \not\equiv \mta_0 \cdot #1'}
\newcommand{\equivz}[1][t]{#1 \equiv \mta_0 \cdot #1'}
\newcommand{\sumgen}[4]{\sum_{#1=1}^{#2} #3_{#1} \cdot #4_{#1}}
\newcommand{\sumdef}[1][t]{\sumgen i n \alpha {#1}}
\newcommand{\candef}[1][t]{\sumdef[#1]\in \can}
\newcommand{\sumdefb}[1]{\sumgen j m \beta {#1}}
\newcommand{\candefb}[1]{\sumdefb{#1}\in \can}

\newcommand{\qbit}{\mathtt{Qbit}}
\newcommand{\nat}{\mathtt{nat}}
\newcommand{\typelist}[1][T]{\mathtt{List}(#1)}
\newcommand{\sign}[1]{\mathrm{type}(#1)}
\newcommand{\signdef}[1][\mtb]{\sign{#1} = T_1 \times \dots \times T_n \to T}
\newcommand{\typedom}[1]{\mathrm{Dom}(#1)}
\newcommand{\subcont}[1][\Gamma]{\mathrm{Sub_0}(#1)}
\newcommand{\typing}[4]{
  \ifx&#1&\varnothing \else #1\fi;
  \ifx&#2&\varnothing \else #2\fi
  \vdash_{\mfR} #3 : #4
}
\newcommand{\typingamp}[2]{
  #1 \Vdash_{\mfR} #2
}
\newcommand{\typingdef}[1][t]{\typing \Gamma \Delta {#1} T}
\newcommand{\typingclosed}[2]{\vdash_{\mfR} #1 : #2}
\newcommand{\kronweird}[2]{\delta_\perp(#1,#2)}
\newcommand{\ortho}{\perp}
\newcommand{\orthored}{\ortho_\mfR}

\newcommand{\funcin}{\mathrm{In}(\mtf)}
\newcommand{\symbin}{\mathrm{In}(\mtb)}
\newcommand{\struct}[1]{\mathtt{struct}(#1)}
\newcommand{\structset}[1]{\mathtt{S}(#1)}
\newcommand{\qtrscirc}{\mathtt{QTRSCirc}}
\newcommand{\rank}[1][\mtf]{\mathrm{rk}(#1)}
\newcommand{\qtrsrec}{\mathtt{SRec}}
\newcommand{\timeset}[1][\tau]{\mathtt{Time}(#1)}
\newcommand{\qtrstime}[1][\tau]{\mathtt{QTRS}(#1)}

\newcommand{\qtrspoly}{\mathtt{QTRS}(Poly)}
\newcommand{\famcirc}{\mathtt{C}(\mfR)}

\newcommand{\orddef}[2]{
  \expandafter\xdef\csname ord#1\endcsname{\succ_{\noexpand #2}}
  \expandafter\xdef\csname ord#1eq\endcsname{\succcurlyeq_{\noexpand #2}}
}
\orddef{}{}
\orddef{nat}{\mbN}
\orddef{ext}{w}
\orddef{func}{\mfR}
\orddef{rpo}{rpo}
\orddef{pol}{\mbN[X]}
\orddef{assign}{\assign -}

\newcommand{\assign}[1]{\llparenthesis #1 \rrparenthesis}
\newcommand{\polset}[1]{\mbN[X_1, \dots, X_{#1}]}
\newcommand{\polvar}{\mathrm{VAR}}
\newcommand{\dpair}[2]{\langle #1, #2 \rangle}
\newcommand{\dpairdef}[1]{\dpair{s_{#1}}{t_{#1}}}
\newcommand{\subcan}{\,\triangleleft_{\,\can}\,}
\newcommand{\qt}[2]{\theta_{#1}(#2)}

\nolinenumbers

\title{Quantum Term Rewrite Systems: Applications to Complexity Analysis}

\author{Kostia Chardonnet}{Université de Lorraine, CNRS, Inria, LORIA, F-54000 Nancy, France}{kostia.chardonnet@inria.fr}{0009-0000-0671-6390}{}
\author{Emmanuel Hainry}{Université de Lorraine, CNRS, Inria, LORIA, F-54000 Nancy, France}{emmanuel.hainry@loria.fr}{0000-0002-975w0-0460}{}
\author{Romain Péchoux}{Université de Lorraine, CNRS, Inria, LORIA, F-54000 Nancy, France}{romain.pechoux@loria.fr}{0000-0003-0601-5425}{}
\author{Thomas Vinet}{Université de Lorraine, CNRS, Inria, LORIA, F-54000 Nancy, France}{thomas.vinet@inria.fr}{0009-0007-8547-6145}{}
\authorrunning{K. Chardonnet, E. Hainry, R. Péchoux, and T. Vinet}

\hideLIPIcs

\ccsdesc[500]{Theory of computation~Rewrite systems}
\ccsdesc[500]{Theory of computation~Quantum complexity theory}

\keywords{Term Rewrite Systems, Quantum Computing, Resource Analysis}

\begin{document}

\maketitle

\begin{abstract}
Term Rewrite Systems (TRS) is a computational model offering a level of
abstraction well-suited towards static analysis, e.g., termination or complexity
analyses. In this paper, we introduce Quantum Term Rewrite Systems (QTRS), an
extension of TRS to quantum computing, thus allowing to benefit from quantum
advantage while being able to certify the complexity. We ensure that QTRS
correspond to physically realizable processes and adapt techniques to obtain
termination certificates or generic bounds on the reduction length. We delineate a
class of terminating QTRS that can be compiled to uniform families of quantum
circuits of size bounded by the reduction length. Conversely, this class is
universal for quantum circuits. In particular, we show a characterization of the
class of functions computable in quantum polynomial time, known as $\fbqp$. 
\end{abstract}

\section{Introduction}
\label{sec:intro}
\subparagraph*{Motivations. }

The past few decades saw the emergence of quantum computation, a paradigm
in which problems can be solved more efficiently compared to their
classical counterparts~\cite{Sho97}.
This paradigm has been studied through multiple computational models, e.g.,
quantum circuits~\cite{Yao93,NC12}, linear optics~\cite{KLM01}, and
ZX-calculus~\cite{CD11}.
Recently, a particular attention has been paid to the development of
higher-level models, namely quantum programming languages.
Quantum programs can be classified in two categories, depending on their
control flow.
A first approach, known as \emph{classical control}, considered that only
classical data, e.g., the outcome of a measurement, could influence the control
flow~\cite{SV09,GLRSV13,QASM17,FY21}.
On the other hand, \emph{quantum control} allows applying controlled
operations depending on quantum data, yielding superpositions of programs.
Quantum control gained interest as it provides a computational advantage over
classical control~\cite{ACB14,TCM+21,KOY24}, and has been approached through
various languages~\cite{SVV18, DCM22, YVC24, BPP26}.
Thus, a natural issue is to develop and study \emph{hybrid languages}, i.e.,
languages that feature both classical and quantum control
flow/data~\cite{Qunity23, Yin24, DLPZ25, CI26, CHPV26}.

The clear theoretical advantage of quantum computation is however mitigated by
the hardness of hardware implementation.
In particular, each type of hardware has its own set of constraints, and
quantum gates may be more or less costly to implement depending on the
hardware.
Thus, providing a static analysis of the (low-level) resources used by a
program (e.g., the size, depth, or gate count of its corresponding quantum
circuit) is a relevant issue, and has been studied in~\cite{DLMZ10, AMPPZ22,
HPS23, CDL25, FHPS25}.
These works fall within the field of \emph{Implicit Computational Complexity}
(ICC), by providing a language that characterizes exactly a well-known
complexity class.
However, these languages are targeted towards a specific complexity class
through typing or syntactic restrictions, and are thus not a good match for a
generic complexity analysis.

With regard to classical programs, Term Rewriting Systems (TRS)~\cite{BN98} are
a computational model that is sufficiently abstract and simple to allow for
automatic and semi-automatic formal analyses of programs. In particular,
numerous generic techniques exist for analyzing their termination (e.g.,
\cite{AG00,ADLY20}) as well as their runtime and complexity, (e.g.,
\cite{BMM11, AM13}), making them well-suited for studying ICC-related
properties.
TRS have also proven to be a robust computational model that allows for simple
and natural extensions to other paradigms, such as probabilistic
rewriting~\cite{BK02, ADLY20} or higher-order rewriting~\cite{MN98, Yam01}.

The development of an extension to the quantum setting is still missing,
which would allow to reuse existing techniques, to certify properties
on the quantum resources.

\subparagraph*{Contributions. }

This paper fills the gap by introducing the notion of \emph{quantum term rewrite
systems} (QTRS).
One advantage of QTRS, in contrast with aforementioned resource-aware
languages, is that QTRS feature quantum control, and have no fixed initial set
of quantum gates, making them prone to a general resource analysis.
Furthermore, for an expressive fragment, QTRS correspond to quantum circuits,
and their size can be related with the runtime-complexity of the QTRS.
This paves the way for the reuse of any complexity technique on standard TRS.
Our work contains the following contributions:

\begin{itemize}
  \item We introduce QTRS as an extension of TRS, where the semantics
    achieves quantum parallelism~\cite{NC12}, and the type system ensures the
    physicality of the terms.
    Standard properties are also proven, such as confluence
    (Lemma~\ref{lem:confluence}), subject reduction (Lemma~\ref{lem:subred}),
    and a characterization of the normal forms (Lemma~\ref{lem:normal-forms}).
  \item We show that type inference is $\Pi_2^0$-hard, thus undecidable
    in the general case (Theorem~\ref{thm:undecide}).
    However, decidability can be recovered for an expressive subset of terms,
    on which type inference is decidable in polynomial time
    (Theorem~\ref{thm:type-inference}).
  \item We exhibit a fragment of terminating QTRS, which is universal
    for quantum circuits (Theorem~\ref{thm:universality}). 
    In this fragment, QTRS can be compiled into uniform families of quantum
    circuits (Theorem~\ref{thm:compile}).
    Furthermore, for QTRS with a specific recursive structure, the size of
    the circuit can be bounded by the runtime-complexity of the QTRS
    (Theorem~\ref{thm:compile-bound}).
    In particular, compilable QTRS terminating in polynomial time characterize
    exactly the class $\fbqp$ of functions computable in quantum
    polynomial time (Theorem~\ref{thm:fbqp}).
  \item We study how complexity and termination properties can be inferred on
    QTRS by extending any ordering on standard TRS to the QTRS setting
    (Section~\ref{sec:wpo}).
    We show that existing techniques --- e.g., polynomial
    interpretations~\cite{Lan79} (Theorem~\ref{thm:inter}), and dependency
    pairs~\cite{AG00} (Theorem~\ref{thm:dpair}) --- can be adapted and reused
    on QTRS to guarantee termination and complexity properties of a QTRS.
\end{itemize}

\subparagraph*{Related work. }

\newcommand{\foq}{\mathtt{FOQ}}
Quantum complexity properties have already been studied for various languages
with classical control:
\cite{DLMZ10} characterizes (polynomial time) complexity classes on a variant
of quantum lambda calculus;
\cite{AMPPZ22} adapts expectation transformers to a quantum imperative language
to infer expected costs and values and has been implemented in~\cite{MS26};
\cite{CDL25} develops a dependent type system on a variant of the Quipper
circuit description language~\cite{GLRSV13}.

With regard to quantum control, $\foq$~\cite{HPS23} is a hybrid imperative
language characterizing functions terminating in quantum polynomial time, and
has been refined in \cite{FHPS25} to obtain a characterization of quantum
polylogarithmic time.
However, only qubit lists are featured, while QTRS can express all standard
inductive datatypes.
Furthermore, $\foq$ does not feature general recursion, which is necessary to
achieve a fully generic resource analysis.

TRS have been widely studied in order to guarantee ICC-related properties.
Here we provide a non-exhaustive overview.
One panel of studies is obtained through \emph{interpretations}, yielding
characterization of polynomial complexity~\cite{BCMT01}.
\emph{Quasi-interpretations} relax strict
monotonicity conditions by incorporating path orderings~\cite{Der82}, yielding
polynomial time and space characterization~\cite{BMM11}.
Another set of techniques reside in orderings, which yield polynomial
termination via light multiset path order~\cite{Mar03} or polynomial path
orders~\cite{AM13}, while polynomial space can be obtained via Knuth-Bendix
orders~\cite{BM10}.
\emph{Dependency pairs}~\cite{AG00}, originally a tool for studying
termination, have also been
refined to analyze the runtime complexity of the TRS~\cite{HM08, AM09, EGN13}.

All the aforementioned techniques  have also been extended to other paradigms.
For example, termination has been studied for probabilistic TRS~\cite{BG05,
KG23}, conditional TRS~\cite{LM17}, and higher-order TRS~\cite{MN25}.
Similarly, runtime complexity analysis has been performed
on conditional~\cite{KMS17} and higher-order TRS~\cite{KV23}, and for
TRS with a parallel reduction strategy~\cite{BFG24}.
In particular, characterizations of basic feasible functions~\cite{CK89} have
been obtained for higher-order TRS~\cite{HP20, BDLKV25}.
While these complexity results are often aimed at achieving upper bounds, some
work has also been carried to obtain lower bounds~\cite{AFGHS16} or to ensure
constant runtime~\cite{FG18}.
We believe that the very general concept of a quantum TRS presented
in this paper can be adapted to these various paradigms.

\section{Quantum Term Rewrite Systems}
\label{sec:qtrs}
This section introduces \emph{Superposed TRS} (STRS) as an extension of
standard TRS~\cite{BN98} that features superpositions.
The semantics of STRS is defined in term of a call-by-basis reduction
strategy~\cite{DCGMV19} mimicking quantum parallelism~\cite{NC12}.
\emph{Quantum TRS} are then introduced as a typed restriction
ensuring physicality.

\subsection{Syntax}\label{ss:syntax}

Let $\mcX$ be a countable set of variables $\mtx, \mty, \dots$, and let $\mcA$,
$\mcC$, $\mcF$ be three signatures containing \emph{amplitude
symbols} $\mta$, \emph{constructor symbols} $\mtc$, and \emph{function symbols}
$\mtf$, respectively.
Recall that a \emph{signature} is a set of symbols $\mtb$, together with a
fixed \emph{arity} $\arity \in \mbN$.
The four sets are assumed to be disjoint.

First-order TRS are extended by allowing \emph{term superposition}, that is,
having a binary sum for superpositions $(+)$ and multiplication by an
amplitude ($\alpha \cdot$). This is formalized by the following grammar.
\begingroup
\renewcommand{\arraystretch}{1.1}
\[
  \begin{array}{r r @{\quad \Coloneqq \quad} l}
    (\text{Natural numbers}) & \iota & \mtx \gmid \natz \gmid \nats(\iota) \\
    (\text{Amplitudes}) & \alpha &\mta(\iota, \dots,
    \iota)  \gmid \alpha \cdot \alpha \gmid \alpha + \alpha \\
    (\text{Terms}) & \trset \ni\, t & \mtx \gmid \mtc(t, \dots, t) \gmid \mtf(t,
    \dots, t) \gmid \alpha \cdot t \gmid t + t \\
    (\text{Values}) & \trsetval \ni v & \mtc(v, \dots, v) \gmid
    \underline \alpha \cdot v
    \gmid v + v
  \end{array}
\]
\endgroup

The distinction between constructor and function symbols is akin to
\emph{constructor TRS}, i.e., function symbols in $\mcF$ will be evaluated
through rewrite rules, while constructor symbols in $\mcC$ are constants.
In the following, we consider that $\mcC$ contains a unit constructor $()$;
qubit constructor symbols $\zket, \oket$; constructors for natural numbers
$\natz$ and $\nats$; a pair constructor $\consp$ of arity $2$, written $(a, b)
\triangleq \consp(a, b) $; as well as list constructors, $\nil$ of arity $0$,
and $\consht$ of arity $2$, written $h :: t \triangleq \consht(h, t)$.

Amplitude symbols of arity $0$ correspond to fixed scalars. In what follows,
$\mta_\lambda$ will be the arity-$0$ amplitude symbol encoding the
scalar $\lambda \in \mbC$.
Amplitude symbols of greater arity will be used to express complex programs
(see Example~\ref{ex:qft}), where the amplitudes may vary.

Let $\var r$ be the set of variables occurring in any term or amplitude $r$.
If $\var r = \emptyset$, $r$ is said to be \emph{ground}, and written
$\underline r$. Let $\trsetground$ denote the set of ground terms.
A term is called \emph{classical} if it has neither superposition,
nor amplitude; denote $\trsetclassical$ as the set of classical terms.
Remark that a classical term can be viewed as a term of a standard TRS.
The \emph{size} of a term $t$, denoted $\size t$, is the number of occurrences
of variables and symbols inside $t$.

A $\mtf$-\emph{rewrite rule}, also called simply \emph{rewrite rule}, is a pair
$l \to r$, with $l,r \in \trset$, such that $l = \mtf(p_1, \dots, p_n)$, where
$p_i$ is a classical term with no function symbol, called \emph{pattern}, and
$\var l \supseteq \var r$.
A set of rewrite rules is said to be \emph{orthogonal}, if the rewrite rules
are \emph{left-linear}, i.e., each variable occurs at most once in the
left-hand side of each rule, and \emph{non-overlapping}, i.e., no pair of
left-hand side of rules can overlap.

\begin{definition}
  A \emph{Superposed Term Rewrite System (STRS)} $\mfR$ is a tuple
  $\trsyntax_\mtf$, where $\mcR$ is a set of orthogonal rewrite rules $l \to r$
  with $l, r \in \trset$, and $\mtf \in \mcF$ is called the \emph{main function
  symbol}.
  The set of rules of $\mfR$ will be denoted by $\trsrule \triangleq \mcR$.
\end{definition}

\begin{example}\label{ex:cliffordt}
  Quantum gates can be written as a STRS, by defining one rule for each action
  on a basis element.
  For example, define the STRS $\mfR$ containing the rewrite rules $\mcR$ below.
  $\mfR$ expresses the Clifford + T set of gates, which is known to
  be \emph{universal}~\cite{NC12}, i.e., any quantum gate on $n$ qubits can be
  approximated by a sequence of gates from this set.
  \[
    \mcR = \left\{
      \begin{array}{r@{\hspace{3pt}} l r@{\hspace{3pt}} l}
        \mathtt{X}(\zket) &\to \oket &
        \mathtt{X}(\oket) &\to \zket \\
        \mathtt{T}(\zket) &\to \zket &
        \mathtt{T}(\oket) &\to \mta_{e^{i \pi / 4}} \cdot \oket \\
        \mathtt{H}(\zket) &\to \stwo \cdot \zket + \stwo \cdot \oket &
        \mathtt{H}(\oket) &\to \stwo \cdot \zket + \mta_{-1}
        \cdot \stwo \cdot \oket \\
        \mathtt{CNOT}(\zket, \mtq) & \to (\zket, \mtq) &
        \mathtt{CNOT}(\oket, \mtq) &\to (\oket, \mathtt{X}(\mtq)) \\
      \end{array}
    \right\}
  \]
\end{example}

\begingroup
\newcommand{\sep}{\\[0.5ex]}
\newcommand{\ctrl}{\mathtt{ctrl}}
\newcommand{\phase}{\mathtt{phase}}
\newcommand{\rot}{\mathtt{rot}}
\newcommand{\rec}{\mathtt{rec}}
\newcommand{\inv}{\mathtt{inv}}
\newcommand{\qft}{\mathtt{qft}}
\begin{example}\label{ex:qft}
  One can encode the Quantum Fourier Transform as a STRS with main function
  symbol $\mathtt{qft}$ and the following set of rules, where $\mta(n)$
  corresponds to the scalar $e^{2i \pi / 2^n}$,
  and $(t_1, \dots, t_n)$ corresponds to $n-1$ pair constructors.
  The second argument of the symbol $\rec$ allows to define
  the behavior of $\rec$ without an additional function symbol.
  \[
    \scalebox{0.9}{$
      \left\{
        \begin{aligned}
          \inv(\nil, l) &\to l &
          \inv(h::t, l) &\to \inv(t, h::l) \sep
          \phase(\zket, n) &\to \zket &
          \phase(\oket, n) &\to \mta(n) \cdot \oket \\
          \ctrl((q, \zket, t, l), n) &\to (q, t, \zket::l) &
          \ctrl((q, \oket, t, l), n) &\to (\phase(q, n), t, \oket::l) \\
          \rot((q, \nil, l), n) &\to q :: \inv(l, \nil) &
          \rot((q, h::t, l), n) &\to \rot(\ctrl((q, h, t, l), n), \nats(n)) \\
          \rec(\nil, b) &\to \nil &
          \rec(h::t, \natz) &\to
          \rec(\rot((\mathtt{Had}(h), t, \nil), \nats(\nats(\natz))),
          \nats(\natz))  \\
          \rec(h::t, \nats(b)) &\to h::\rec(t, b) &
          \qft(l) &\to \inv(\rec(l, \natz), \nil)
        \end{aligned}
      \right\}
    $}
  \]
\end{example}
\endgroup

\subsection{Semantics}\label{sec:semantics}
Before introducing the semantics of a STRS, several steps must be performed.
First, the Hilbert structure of quantum spaces requires us to
consider terms modulo an \textit{equivalence relation} on
the vector space generated by term superpositions.
to ensure reductions only consider meaningful terms.
Third, a specific \textit{evaluation strategy} will be fixed for
classical terms.

\noindent \textit{Vector space structure.}
Defining a vector space structure requires checking properties on amplitudes,
e.g., if an amplitude corresponds semantically to $0$.
To that end, each amplitude symbol $\mta \in \mcA$ comes with a total function
$\inter \mta : \mbN^{\arity[\mta]} \to \mbC$ called \emph{interpretation}.
The interpretation is then defined on ground amplitudes by canonical extension:
\begin{align*}
  \inter{\natz} &\triangleq 0  &
  \inter{\nats(\underline \iota)} &\triangleq \inter{\underline \iota} + 1 &
  \inter{\mta(\underline{\iota_1},\dots, \underline{\iota_n})} &\triangleq
  \inter\mta(\inter{\underline{\iota_1}}, \dots, \inter{\underline{\iota_n}}) \\
  \inter{\underline{\alpha_1} + \underline{\alpha_2}} &\triangleq
  \inter{\underline{\alpha_1}} + \inter{\underline{\alpha_2}} &
  \inter{\underline{\alpha_1} \cdot \underline{\alpha_2}} &\triangleq
  \inter{\underline{\alpha_1}} \times \inter{\underline{\alpha_2}}
\end{align*}

Recall that a substitution is a map $\sigma: \mcX \to \trsetval$. Let $r
\sigma$ be the application of the substitution $\sigma$ to the  term or
amplitude $r$.
The interpretation can be extended to arbitrary amplitudes as follows:
define $\inter \alpha \triangleq \lambda$ if for any substitution $\sigma$, if
$\alpha \sigma$ is a ground amplitude, then $\inter{\alpha \sigma} = \lambda$.
Two amplitudes $\alpha, \beta$ are \emph{equivalent} if
$\inter{\alpha}= \inter{\beta}$.

\begingroup
\renewcommand{\arraystretch}{1.5}
\begin{figure*}[t]
  \begin{center}
    \begin{prooftree}
      \infer0{t_1 + t_2 \equiv t_2 + t_1 \vphantom{(t_1)}}
    \end{prooftree}
    \qquad
    \begin{prooftree}
      \infer0{t_1 + (t_2 + t_3) \equiv (t_1 + t_2) + t_3}
    \end{prooftree}
    \qquad
    \begin{prooftree}
      \hypo{\inter \alpha = 0}
      \infer1{s + \alpha \cdot t \equiv s\vphantom{(t_1)}}
    \end{prooftree}
    \qquad
    \begin{prooftree}
      \hypo{\inter \alpha = 1}
      \infer1{\alpha \cdot t \equiv t\vphantom{(t_1)}}
    \end{prooftree}
    \\
    \begin{prooftree}
      \hypo{\phantom{t}}
      \infer1{ \alpha \cdot (\beta \cdot t) \equiv (\alpha \cdot \beta) \cdot t}
    \end{prooftree}
    \qquad
    \begin{prooftree}
      \hypo{\phantom{t}}
      \infer1{ \alpha \cdot (t_1 + t_2) \equiv \alpha \cdot t_1 +
      \alpha \cdot t_2}
    \end{prooftree}
    \qquad
    \begin{prooftree}
      \hypo{\phantom{t}}
      \infer1{  \alpha \cdot t + \beta \cdot t \equiv (\alpha + \beta) \cdot t}
    \end{prooftree}
    \\
    \begin{prooftree}
      \hypo{\phantom{t}}
      \infer1{  \mtb(\dots, \alpha \cdot t, \dots) \equiv \alpha
      \cdot  \mtb(\dots, t, \dots)}
    \end{prooftree}
    \qquad
    \begin{prooftree}
      \hypo{\phantom{t}}
      \infer1{ \mtb(\dots, t_1 + t_2, \dots) \equiv \mtb(\dots, t_1,
      \dots)+\mtb(\dots, t_2, \dots) }
    \end{prooftree}
    \\[1em]
    \begin{prooftree}
      \hypo{s \equiv t}
      \infer1{C[s] \equiv C[t]}
    \end{prooftree}
  \end{center}
  \caption{Equivalence relation $\equiv$}
  \label{tab:equiv}
\end{figure*}
\endgroup

A ($1$-hole) \emph{context} is a term $C$ containing exactly one
occurrence of the special symbol $\diamond$ and defined by the
following grammar.
\[
  (\text{Contexts}) \qquad C \Coloneqq \, \diamond \gmid \mtc(t_1, \dots, C,
  \dots, t_n) \gmid \mtf(t_1, \dots, C, \dots, t_n) \gmid \alpha \cdot C \gmid
  C + t
\]
Let $C[t]$ be the term obtained by substituting $t$ to  $\diamond$ in $C$.
This notion can be generalized to \emph{$n$-hole contexts} $C[\diamond_1,
\dots, \diamond_n]$, which are terms containing one occurrence of each symbol
$\diamond_1, \dots, \diamond_n$.
A \emph{classical context} is a context with no amplitude ($\alpha \cdot$), nor
superposition ($+$).

\emph{Term equivalence} $\equiv$ is defined as the smallest equivalence
relation over $\trset$ obtained from the rules of Figure~\ref{tab:equiv}.
The two first lines of Figure~\ref{tab:equiv} express the rules of a vector
space structure and allow us to define unambiguously general summation as
$\sum_{i=1}^n \alpha_i \cdot t_i \triangleq \alpha_1 \cdot t_1 + (\dots +
\alpha_n \cdot t_n)$.
The remaining lines of Figure~\ref{tab:equiv} highlight the
linearity of quantum programs.

\noindent \textit{Canonical forms.}
Quantum parallelism will be defined using a restricted form of superposition to
avoid non-physical reductions, e.g., reducing a term of amplitude $0$. This is
formalized below by the notion of \emph{canonical form}.

\begin{definition}\label{def:canonical}
  A \emph{canonical form} of a term $t$ is any term $\sumdef$ such that $t
  \equiv \sumdef$, where $t_i$ are pairwise distinct classical terms and
  $\inter{\alpha_i} \neq 0$. The set of canonical forms is denoted by $\can$.
\end{definition}

\noindent \textit{Evaluation strategy.}
Finally, the reduction of classical terms will be defined by fixing a given
evaluation strategy. This is performed with the help of
\emph{evaluation contexts}.
\[
  (\text{Eval. contexts}) \quad
  E \Coloneqq \diamond
  \gmid \mtc({v_1}, \dots, {v_{i-1}}, E, t_{i+1}, \dots, t_n)
  \gmid \mtf({v_1}, \dots, {v_{i-1}}, E, t_{i+1}, \dots, t_n)
\]

\begin{figure*}[t]
  \centering
  \begin{prooftree}
    \hypo{s \equiv s'}
    \hypo{s' \toq t'}
    \hypo{t' \equiv t}
    \infer3[(E)]{s \toq t \vphantom{\sum_{i=1}^n}}
  \end{prooftree}
  \quad
  \begin{prooftree}
    \hypo{\candef[s]}
    \hypo{\forall i,\ s_i \toqif t_i}
    \hypo{\exists i,\ s_i \toq t_i}
    \infer3[(Q)]{\sumdef[s] \toq \sumdef[t]}
  \end{prooftree}
  \\[2ex]
  \begin{prooftree}
    \hypo{l \to r \in \trsrule}
    \hypo{E[l \sigma] \in \trsetclassical}
    \infer2[(C)]{E[l\sigma] \toq E[r\sigma]}
  \end{prooftree}
  \vspace{10pt}
  \caption{
    Inference rules of the relation $\toq\ \subseteq \trsetground \times
    \trsetground$
  }
  \label{tab:toq}
\end{figure*}

The semantics of a STRS $\mfR$ is then defined inductively in
Figure~\ref{tab:toq} as the relation $\toq$ on ground terms.
In Rule (Q), $\toq$ relies on the intermediate relation $\toqif$, defined by
\[
  s \toqif t \iff (s \in \nf \wedge t = s) \vee s \toq t
\]
where $\nf \triangleq \{s \in \trsetground \gmid \nexists\,t,\, s \toq t\}$
is the set of \emph{normal forms}.

We will now provide some insights into the rules of Figure~\ref{tab:toq}.
Rule (E) allows one to perform reduction with respect to the
equivalence relation on the vector space.
Rule (Q) reduces superpositions only if expressed as canonical forms, ensuring
that only meaningful reductions are fired.
Furthermore, $\toqif$ implies that any term that is able to reduce will reduce.
Therefore, $\toq$ achieves quantum parallelism~\cite{NC12}, similarly
to~\cite{GMM20}.
An alternative reduction strategy not implementing parallelism would correspond
to classical simulation, leading any program acting on $n$ qubits to be
evaluated in $2^n$ reduction steps.
Finally, Rule (C) considers classical terms: variables are
substituted by classical values, thus $\toq$ follows a \emph{call-by-basis}
strategy~\cite{DCGMV19}. The definition of evaluation contexts also
imposes an evaluation of terms from left to right.
While such choice is arbitrary, having a fixed strategy is required to obtain
complexity results in Section~\ref{sec:complexity}.

Given a STRS $\mfR$, define $\toqs$ as the reflexive and transitive closure of
$\toq$.
We say that $t$ \emph{terminates in at most $k$ steps}, and write $t
\toq^{\leq k} s$, if all chains of reduction starting from $t$ reach a normal
form in at most $k$ steps.
A term \emph{terminates} if it terminates in at most $k$ steps for
some $k \in \mbN$. A STRS $\mfR$ is said to be \emph{terminating}, if any term
$t \in \trsetground$ terminates.

\begingroup
\newcommand{\sgn}{\mathtt{sgn}}
\begin{example}\label{ex:had-gen}
  Define the unitary matrix $Q_n$, representing a single qubit gate,
  for any $n \in \mbN$, along the rewrite rules of its corresponding STRS:
  \[
    Q_n \triangleq \fractwo \cdot
    \begin{pmatrix}
      1 & e^{i \pi / n} \\
      e^{-i \pi / n} & -1
    \end{pmatrix}
    \qquad
    \mcR \triangleq \left\{
      \begin{aligned}
        \mtf(n, \zket) &\to \sgn^+ \cdot \zket + \mta^+(n) \cdot \oket \\
        \mtf(n, \oket) &\to \mta^-(n) \cdot \zket + \sgn^- \cdot \oket
      \end{aligned}
    \right\}
  \]
  where $ \inter{\sgn^\pm} \triangleq \pm \fractwo$ and
  $\inter{\mta^\pm}(x) \triangleq \fractwo e^{\pm i \pi / x}$.
  One can verify that applying twice $\mtf(n, \cdot)$ is equivalent to the
  identity, e.g., if applied to $\zket$, for any ground natural number
  $\underline{\iota}$:
  \[
    \begin{aligned}
      \mtf(\underline{\iota}, \mtf(\underline{\iota}, \zket))
      &\toq \mtf(\underline{\iota}, \sgn^+ \cdot \zket +
      \mta^+(\underline{\iota}) \cdot \oket)
      \equiv \sgn^+ \cdot \mtf(\underline{\iota}, \zket) +
      \mta^+(\underline{\iota}) \cdot \mtf(\underline{\iota}, \oket) \\
      &\toq \left(\sgn^+ \cdot \sgn^+ + \mta^+(\underline{\iota}) \cdot
      \mta^-(\underline{\iota}) \right) \cdot \zket +
      \left(\sgn^+ \cdot \mta^+(\underline{\iota}) +
        \mta^+(\underline{\iota}) \cdot \sgn^-
      \right) \oket \\
      &\equiv \zket
    \end{aligned}
  \]
  The first and second reductions respectively used the rules (C)
  and (Q) from Figure~\ref{tab:toq}.
  The last equivalence can be obtained by seeing the reduced term as
  $\alpha \cdot \zket + \beta \cdot \oket$, and checking that $\inter \alpha =
  1$ and $\inter \beta = 0$, e.g., see below for $\inter \alpha = 1$:
  \[
    \begin{aligned}
      \inter{\alpha} &= \inter{\sgn^+ \cdot \sgn^+ + \mta^+(\underline{\iota})
      \cdot \mta^-(\underline{\iota})}
      &= \inter{\sgn^+}\times \inter{\sgn^+} +
      \inter{\mta^+}(\underline{\iota}) \times
      \inter{\mta^-}(\underline{\iota}) =1\\
    \end{aligned}
  \]
  Therefore, using rule (E) from Figure~\ref{tab:toq}, $\mtf(\underline{\iota},
  \mtf(\underline{\iota}, \zket))$ rewrites in two steps to $\zket$.
\end{example}
\endgroup

\subsection{Type System for Physicality}

STRS extend standard TRS by allowing any general term superposition.
However, quantum programs need to satisfy additional properties to be physically
sound, e.g., linear use of data or norm preservation.
This is tackled by introducing a type system, both on terms and amplitudes,
allowing us to then define Quantum TRS (QTRS) as the STRS with a
physical reality.

Towards that end, a unique \emph{typed signature} $\sign \mtb \triangleq T_1
\times \dots \times T_{\arity} \to T$ is assigned to each symbol $\mtb \in
\symbols$, where $T, T_1,\ldots,T_{\arity}$ are taken from a countable set of
basic types. The types $\mathbb 1$ for unit, $\qbit$ for qubits, or $\nat$ for
natural numbers are examples of basic types.
If $\mtb$ is of arity $0$, its signature is written simply $\sign \mtb = T$.
The constructor symbols introduced in Section~\ref{ss:syntax} can be
equipped with the following signatures, for any basic type $T, T_1, T_2$:
\[
  \begin{aligned}
    \sign{()} &\triangleq \mathbb 1 &
    \sign{\zket} = \sign{\oket} &\triangleq \qbit \\
    \sign{\natz} &\triangleq \nat &
    \sign{\nats} &\triangleq \nat \to \nat \\
    \sign{\nil_T} &\triangleq \typelist &
    \sign{\consht_{\typelist} } &\triangleq T \times \typelist \to \typelist \\
    \sign{\consp_{T_1, T_2}} &\triangleq T_1 \times T_2 \to \times_{T_1, T_2}&&
  \end{aligned}
\]
Note that polymorphic types (e.g., lists or pairs) need to have multiple
constructors defined (one for each possible input type).
However, when clear from the context, the previous notations are used,
e.g., $h::t$ and $(h,t)$.
To ensure linearity of quantum data, the
set of basic types is split disjointly between \emph{quantum types} $Q \in \mtQ$
and \emph{classical types} $C \notin \mtQ$. The set of quantum types
$\mtQ$ is defined below:
\[
  \mtQ \triangleq \set{\qbit} \cup \set{
    T \gmid \exists\,\mtc \in \mcC,\exists\,i \in \mbN,
    \signdef[\mtc] \wedge T_i \in \mtQ
  }
\]
A quantum type is either the qubit type $\qbit$, or any type having (at least)
one constructor symbol with a quantum type in its signature.

\begingroup
\renewcommand{\arraystretch}{3}
\begin{figure*}[t]
  \[
    \centering
    \begin{array}{c}
      \begin{prooftree}
        \hypo{\forall i,\ \typing \Gamma{}{\iota_i} \nat}
        \hypo{\mta \in \trsamp}
        \hypo{\arity[\mta] = n}
        \infer3{\typingamp \Gamma {\mta(\iota_1, \dots, \iota_n)}}
      \end{prooftree}
      \qquad
      \begin{prooftree}
        \hypo{\typingamp \Gamma {\alpha_1}}
        \hypo{\typingamp \Gamma {\alpha_2}}
        \infer2{\typingamp \Gamma {\alpha_1 + \alpha_2} \vphantom{(\alpha_i^c)}}
      \end{prooftree}
      \qquad
      \begin{prooftree}
        \hypo{\typingamp \Gamma {\alpha_1}}
        \hypo{\typingamp \Gamma {\alpha_2}}
        \infer2{\typingamp \Gamma {\alpha_1 \cdot \alpha_2}
        \vphantom{(\alpha_i^c)}}
      \end{prooftree}
      \\[1ex]
      \begin{prooftree}
        \hypo{x \in \trsvar}
        \infer1{\typing{\Gamma, x : C}{} x C}
      \end{prooftree}
      \qquad
      \begin{prooftree}
        \hypo{x \in \trsvar}
        \infer1{\typing{\Gamma}{x : Q} x Q}
      \end{prooftree}
      \qquad
      \begin{prooftree}
        \hypo{\typing \Gamma \Delta s T}
        \hypo{s \equiv t}
        \hypo{\fv t \subseteq \typedom{\Gamma \cup \Delta}}
        \infer3{\typing \Gamma \Delta t T}
      \end{prooftree}
      \\
      \begin{prooftree}
        \hypo{\signdef}
        \hypo{\forall i,\ \typing{\Gamma}{\Delta_i}{t_i}{T_i}}
        \hypo{\mtb \in \trscons \biguplus \trsfunc}
        \infer3{
          \typing \Gamma {\Delta_1, \dots, \Delta_n}{\mtb(t_1, \dots, t_n)} T
        }
      \end{prooftree}
      \\
      \begin{prooftree}
        \hypo{\forall i,\ \typing \Gamma \Delta {t_i} Q}
        \hypo{\forall i,\ \typingamp \Gamma {\alpha_i}}
        \hypo{\forall \sigma \in \subcont,\, \sum_{i=1}^n
        \size{\inter{\alpha_i\sigma}}^2 = 1}
        \hypo{\forall i \neq j,\ t_i \orthored t_j}
        \infer4{\typing \Gamma \Delta{\sumdef} Q}
      \end{prooftree}
    \end{array}
  \]
  \caption{Typing rules of terms and amplitudes}
  \label{tab:typing}
\end{figure*}
\endgroup

\emph{Typing contexts} $\Gamma, \Delta$ are sets of the shape $\set{x_1 : T_1,
\dots, x_n: T_n}$, where $x_i$ are pairwise distinct variables, and
$T_i$ are basic types.
The \emph{domain} of a context is defined as its set of variables, i.e.,
$\typedom \Gamma \triangleq \{x \gmid x:  T \in \Gamma\}$.
Whenever we write $\Gamma, \Delta$ or $\Gamma;\Delta$, it is always assumed that
$\Gamma$ and $\Delta$ share no variables, i.e., $\typedom \Gamma \cap \typedom
\Delta = \varnothing$.
Given a typing context $\Theta$, we say that a substitution $\sigma$ is a
\emph{$\Theta$-context substitution}, if for any $x : T \in \Theta$, $x \sigma$
is a well-typed value of type $T$.
Let $\subcont[\Theta]$ be the set of $\Theta$-context substitutions.

An \emph{amplitude typing judgment} is written $\typingamp \Gamma \alpha$,
indicating that $\alpha$ is a well-typed amplitude, with respect to $\mfR$,
under (non-linear) context $\Gamma$.
A \emph{term typing judgment} is written $\typingdef$, indicating that $t$ is
well-typed with respect to $\mfR$ and basic type $T$, under non-linear context
$\Gamma$ and linear context $\Delta$; when $\Gamma = \Delta = \varnothing$, we
may write $\typingclosed t T$.

These judgments are derived inductively following the
typing rules of Figure~\ref{tab:typing}, for a given STRS $\mfR$.

A term typing judgment aims to type normalized terms only.
Towards that end, an \emph{orthogonality predicate} on terms is introduced.
It relies on the notion of \emph{orthogonal Kronecker product}, which is a
partial map on classical terms $\delta_\perp: \trsetclassical \times
\trsetclassical \to \set{0,1}$, where $\kronweird t t \triangleq 1$, and
$\kronweird s t \triangleq 0$ if there exist a $n+1$-hole context $C$, $q \in
\{0,1\}$, terms $s_1, \dots, s_n$ and $t_1, \dots, t_n$ such that $s = C[\ket
q, s_1, \dots, s_n]$ and $t = C[\ket{1 - q}, t_1, \dots, t_n]$.

\begin{definition}[Orthogonality]\label{def:orthogonality}
  Let $\mfR$ be a STRS.
  Let $\typing \Gamma \Delta s T$ and $\typing \Gamma \Delta t T$ be two terms.
  We say that $s$ and $t$ are \emph{orthogonal}, written $s \orthored t$, if for
  any substitution $\sigma \in \subcont[\Gamma \cup \Delta]$, $s
  \sigma \toqs \candef[v]$, $t \sigma \toqs \candefb w$, and:
  \[
    \sum_{i=1}^n \sum_{j=1}^m \inter{\alpha_i} \inter{\beta_j}^*
    \kronweird{v_i}{w_j} = 0
  \]
  where, given $\lambda \in \mathbb{C}$, $\lambda^*$ is the complex
  conjugate of  $\lambda$.
\end{definition}

Note that the above definition requires $\delta_\perp$ to be defined
for all $(v_i, w_j)$.
Typing a superposition requires the amplitudes to correspond
semantically to a normalized complex vector, and that the normal forms of the
terms are pairwise orthogonal.

Given a symbol $\mtb \in \trscons \biguplus \trsfunc$, define
$\symbin$ as the set of \emph{well-typed inputs of
$\mtb$}, i.e., $\overline{v}  \triangleq (v_1, \dots, v_n) \in \symbin$ if
$\signdef[\mtb]$ and $\typingclosed{v_i}{T_i}$, for all $i$.

Define $\overline v \orthored \overline w$ if $v_i \orthored w_i$ for some $i$.
While typing imposes that terms are normalized, we still need to ensure that
quantum programs preserve the norm, i.e., are \emph{isometries}.

\begin{definition}
  Given a STRS $\mfR$ and a function symbol $\mtf \in \trsfunc$, we say that
  $\mtf$ is an \emph{isometry}, if for all inputs $\overline v, \overline w \in
  \funcin$, $\overline v \orthored \overline w$ implies
  $\mtf(\overline v) \orthored \mtf(\overline w)$.
\end{definition}

Finally, quantum term rewrite systems are defined as STRS whose rewrite rules
are well-typed, and where function symbols are isometries.

\begin{definition}\label{def:qtrs}
  A STRS $\mfR$ is called a \emph{Quantum Term Rewrite System (QTRS)} if all
  function symbols are isometries and for any rule $l \to r \in
  \trsrule$, typing contexts $\Gamma, \Delta$, and type $T$:
  \[
    \typingdef[l] \implies \typingdef[r]
  \]
\end{definition}

\begin{example}
  The STRS from Example~\ref{ex:had-gen} can be showed to be a QTRS.
  Indeed, each side of the rewrite rule can be typed identically.
  Furthermore, as both rewrite rules of $\mtf$ reduce to
  orthogonal terms, $\mtf$ can be shown to be an isometry.
  This is formalized in Appendix~\ref{app:additional}.
\end{example}

\section{Main Results}
\label{sec:results}
This section is devoted to showing that QTRS are well-behaved and
enjoy standard properties (confluence, subject reduction, ...).
It also discusses the hardness of type inference, which is undecidable in
general, and shows that decidability can be recovered under some slight
restrictions.
Finally, we exhibit an expressive fragment of QTRS, which has the same
computational power as quantum circuits: it is universal, and can be compiled
to quantum circuits. Based on this fragment, we characterize the class of
functions computable in quantum polynomial time, known as $\fbqp$~\cite{BV97}.

\subsection{Standard Properties of QTRS}
\label{sec:properties}

Because of orthogonality and the evaluation strategy, STRS are
confluent up to equivalence.

\begin{restatable}[Confluence]{lemma}{confluence}\label{lem:confluence}
  Let $\mfR$ be a STRS, and let $t$ be a term.
  Suppose that there exist two terms $t_1, t_2$ such that $t \toqs t_1$ and $t
  \toqs t_2$. Then, there exist two terms $t_3, t_4$ such that $t_1 \toqs t_3$,
  $t_2 \toqs t_4$ and $t_3 \equiv t_4$.
\end{restatable}

Typing is preserved through reduction for terminating or classical terms.

\begin{restatable}[Subject reduction]{lemma}{subred}\label{lem:subred}
  Let $\mfR$ be a QTRS, and let $\typingclosed s T$ be a well-typed term that
  is either terminating or classical.
  If $s \toq t$, then $\typingclosed t T$.
\end{restatable}

\begin{remark}
  Termination is required to type superpositions obtained after reduction.
  For example, take the QTRS made of the three following rewrite rules:
  \[
    \mcR = \left\{
      \Omega(\natz) \to \Omega(\nats(\natz)) \qquad
      \Omega(\nats(n)) \to \Omega(n) \qquad
      \mtf(\ket i) \to \ket i
    \right\}
  \]
  Take $t = (\stwo \cdot \zket + \stwo \cdot \mtf(\oket), \Omega(\natz))$,
  with $\inter{\stwo} = \fractwo$.
  While $t$ is well-typed, $t \toq \stwo \cdot (\zket, \Omega(\nats(\natz))) +
  \stwo \cdot (\oket, \Omega(\natz))$, which cannot be typed, as summands do
  not terminate.
\end{remark}

A well-typed term always possesses a unique canonical form.

\begin{restatable}[Canonical form]{lemma}{canonicaltyped}
  \label{lem:canonical-typed}
  Let $\mfR$ be a QTRS, and let $\typingdef$ be a well-typed term.
  Then, $t$ has a unique canonical form $\sumdef$ up to reordering and
  amplitude equivalence, and $\typingdef[t_i]$.
\end{restatable}

As a consequence of Lemma~\ref{lem:confluence} and~\ref{lem:canonical-typed},
the orthogonality predicate (Definition~\ref{def:orthogonality})
is sound.

A QTRS $\mfR$ is said to be \emph{total}, if for any function symbol $\mtf$ and
any input $\overline v \in \funcin$, there exist a substitution $\sigma$ and a
rule $l \to r \in \trsrule$ such that $l\sigma = \mtf(\overline v)$.
For total QTRS, normal forms coincide with terms equivalent to values.

\begin{restatable}[Normal forms]{lemma}{progress}\label{lem:normal-forms}
  Let $\mfR$ be a total QTRS. Then, $\nf$ is exactly the set of well-typed
  ground terms equivalent to values.
\end{restatable}

\subsection{Type Inference}\label{sec:inference}

This section discusses the decidability of type inference.
The typing rule for superpositions requires computing additions,
multiplications, and nullity checks on complex numbers.
Recall that algebraic numbers are complex numbers that are roots
of a polynomial in $\mathbb Q[X]$, and write their sets as $\bar{\mathbb C}$.
In the field of algebraic numbers, equality is decidable, and product and sum
are computable~\cite{HHK05}. In this section, we thus restrict ourselves to
amplitudes symbols satisfying
$\inter \mta : \mbN^{\arity[\mta]} \to \bar{\mbC}$.

The typing rule for superpositions (Figure~\ref{tab:typing}) also
requires orthogonality checks $s \orthored t$, which imply checking
termination wrt a universal quantification over all possible
substitutions  (see Definition~\ref{def:orthogonality}). Hence,
typing is undecidable in the general case:
we show that it is at least as hard as the Universal Halting Problem, i.e.,
$\Pi^0_2$-hard in the arithmetical hierarchy~\cite{EGZ09}.

\begin{restatable}[Undecidability of type inference]{theorem}{inferenceundecide}
  \label{thm:undecide}
  Given a STRS $\mfR$, deciding whether $t$ can be typed is $\pitwo$-hard,
  and belongs to $\Sigma_3^0$.
\end{restatable}

Moreover, type inference without equivalence $\equiv$ is $\pitwo$-complete.
Decidability of type inference can be recovered by restricting the typing
rules: a well-typed term $t$ is said to be \emph{easily typed}, if its typing
derivation does not use any equivalence rule and, in any typing rule for
superpositions, the orthogonality check $s \orthored t$ is replaced by the
check $\kronweird s t = 0$.
\emph{Easy type inference} can be decided in polynomial time and
implies that the term is well-typed.

\begin{restatable}[Decidability of easy type inference]{theorem}{typeinference}
  \label{thm:type-inference}
  Let $\mfR$ be a terminating STRS, where $\mcA$ contains only symbols of arity
  $0$. There is a polynomial $P$ such that easy type inference of $t$ is
  decidable in $P(\size t)$, and implies that $t$ is well-typed.
\end{restatable}

For example, all rewrite rules from Example~\ref{ex:cliffordt} can be easily
typed. However, Theorem~\ref{thm:type-inference} still requires termination,
which is known to be undecidable.
Nevertheless, Section~\ref{sec:complexity} will discuss termination techniques
that can be used to obtain a termination certificate, thus giving a
semi-automatized procedure for type inference.

\subsection{QTRS and Quantum Circuits}\label{sec:circuit}

This section exhibits a relation between QTRS and quantum circuits.
Towards that end, a fragment of QTRS is introduced ($\qtrscirc$,
Definition~\ref{def:qtrs-circ}). This fragment is universal for quantum circuits
(Theorem~\ref{thm:universality}), and can be compiled to families of quantum
circuits (Theorem~\ref{thm:compile}).
For programs with a specific recursive structure, the size of the obtained
circuit can be faithfully related with the runtime-complexity of the QTRS
(Theorem~\ref{thm:compile-bound}), allowing to characterize precisely $\fbqp$
(Theorem~\ref{thm:fbqp}), the natural complexity class for quantum
polynomial time computable functions.

While function symbols of a QTRS preserve orthogonality, they
do not always output a coherent number of qubits.
As a counter-example, just consider the following QTRS $\{\mtf(\zket)
  \to \zket :: \nil,
\mtf(\oket) \to \oket :: \oket :: \nil\}$.
Furthermore, a naive padding with ancilla qubits could break
the behaviour of a more general QTRS, e.g., if $\mtf$ is used as an input for
another program.
To tackle this problem, we introduce the notion of \emph{structure}, which
keeps only the classical shape of a value by replacing any occurrence
of a qubit by the unit.
Formally, the structure of a constructor is defined as $\struct
\zket = \struct \oket \triangleq ()$ and $\struct \mtc \triangleq \mtc$ for
$\mtc \in \mcC \setminus \{\zket, \oket\}$.
The structure is then defined as a partial map on values by
\begin{align*}
  \struct{\mtc(\overline{v})} &\triangleq \struct{c}(\struct{\overline{v}}) \\
  \struct {\alpha \cdot v} &\triangleq \struct{v} \\
  \struct{v_1 + v_2 } & \triangleq \struct{v_1}  && \text{if
  }\struct{v_1} = \struct{v_2}
\end{align*}
where $\struct {v_1,\ldots,v_n} \triangleq (\struct{v_1}, \dots,
\struct{v_n})$.
Two classical values have the same structure if they differ only on
their qubit constructors.
The above problem is solved by requiring that $\mtf$ reduces to values of
identical structure, on inputs of identical structure.

\begin{definition}
  Let $\mfR$ be a terminating QTRS.
  We say that $\mtf \in \trsfunc$ is \emph{structure preserving}, if for
  any $\overline{v}, \overline{w} \in \funcin$ such that
  $\struct{\overline{v}} = \struct{\overline{w}}$,
  $\mtf(\overline{v}) \toqs v', \mtf(\overline{w}) \toqs w'$, and
  $\struct{v'} = \struct{w'}$.
\end{definition}

Therefore, compilation of a function symbol $\mtf$ will differ depending
on the considered structure of inputs. In particular, not all rewrite
rules can be applied to a specific structure.
Formally, a \emph{$\mtf$-structural set} $S$ is a set of $\mtf$-rewrite
rules such that $\forall\, l \to r, l' \to r' \in S$,
there exist two substitutions $\sigma, \sigma'$ such that
$l \sigma = \mtf(\overline v)$, $l' \sigma' = \mtf(\overline{v'})$,
and $\struct{\overline v} = \struct{\overline{v'}}$.

Compiling a function symbol $\mtf$ can be done in two ways: either
the considered rewrite rules are viewed as a single unitary gate, or each
rewrite rule is compiled separately as a controlled statement.
We say that $\mtf$ \emph{encodes a unitary}, if any $\mtf$-rewrite rule $l \to
r \in \mcR$ satisfies $r \in \trsetval$, allowing us to compile $\mtf$ as in
the former case.
To properly compile the latter case, it is required that all the considered
rules treat the control qubits identically and leave them untouched.
This property is ensured by the definition below.

\begin{definition}
  Let $\mfR$ be a QTRS, and let $\mtf \in \trsfunc$.
  We say that $\mtf$ \emph{quantum controls}, if for any $\mtf$-structural set
  $S$, there exist a $n$-hole classical context $C$, and a $n+m$-hole
  classical context $C'$ with no function symbol occurrence, such that for
  all $l_j \to r_j \in S$:
  \[
    l_j = C[\ket{i_1}, \dots, \ket{i_n}]
    \quad \wedge \quad
    r_j = C'[\ket{i_1}, \dots, \ket{i_n}, t_1^j, \dots, t_m^j]
  \]
\end{definition}

Finally, we restrict constructor symbols to the set $\mcC_Q \triangleq \{\zket,
\oket, \natz, \nats, \nil, \consht, \consp\}$, so that values have a meaning
circuit-wise.
Altogether, define $\qtrscirc$ as the intersection of all the above
restrictions.

\begin{definition}\label{def:qtrs-circ}
  Define $\qtrscirc$ as the set of QTRS $\mfR$,
  where $\trscons \subseteq \mcC_Q$, and for all $\mtf \in \trsfunc$,
  $\mtf$ is structure preserving, and
  either $\mtf$ encodes a unitary, or $\mtf$ quantum controls.
\end{definition}

\begin{example}
  The QTRS from Example~\ref{ex:cliffordt} belongs to $\qtrscirc$.
  $\mathtt X, \mathtt T, \mathtt H$ all output values, thus encode a unitary,
  while $\mathtt{CNOT}$ leaves the control qubit untouched, thus quantum
  controls.
  All symbols also are structure preserving, as they respectively output a
  qubit and a pair of qubits no matter their inputs.
  Similarly, each function symbol from Example~\ref{ex:qft} quantum controls
  and is structure preserving, as it outputs a qubit list of same size as its
  input, thus it belongs to $\qtrscirc$.
\end{example}

Let us now discuss the relation between $\qtrscirc$ and quantum circuits.
Recall that a family of quantum circuits $(C_n)_{n \in \mbN}$ is said to be
\emph{uniform}, if there exists a Turing machine which, on input $n$, outputs
the circuit representation of $C_n$.

As term rewrite systems are known to be Turing complete~\cite{HL78}, any
uniform family of circuits can be expressed as a QTRS, with the notation $t_0 =
\natz$ and $t_{n+1} = \nats(t_n)$ for all $n \in \mbN$, and by encoding any
$n$-qubit state $\ket \phi = \sum_{x_1\dots x_n \in \set{0,1}^n} \alpha_{x_1
\dots x_n} \ket{x_1 \dots x_n}$ into the value $t_{\ket \phi}$ defined below:
\[
  t_{\ket \phi} \triangleq \sum_{x_1\dots x_n \in \set{0,1}^n}
  \mta_{x_1 \dots x_n} \cdot \ket{x_1}::\dots::\ket{x_n} :: \nil
  ,\qquad \text{with }\inter{\mta_{x_1 \dots x_n}} \triangleq
  \alpha_{x_1 \dots x_n}.
\]

\begin{restatable}[Universality]{theorem}{universality}\label{thm:universality}
  Let $(C_n)_{n \in \mbN}$ be a uniform family of circuits.
  Then, there exists a QTRS $\mfR \in \qtrscirc$ of main symbol $\mtf$ such
  that, for any $n \in \mbN$ and any $n$-qubit state $\ket \phi$, $\mtf(t_n,
  t_{\ket \phi}) \toqs t_{C_n \ket \phi}$.
\end{restatable}

We now prove that any QTRS in $\qtrscirc$ can be approximated by a uniform
family of circuits.
Towards that end, we encode naturally any value $v$ with constructor symbols in
$\mcC_Q$ to a quantum state $\ket v$, by discarding natural bumbers.
This is extended to tuples $\overline{v} =(v_1, \dots, v_n)$ by
$\ket{\overline{v}} \triangleq \ket{v_1} \otimes \dots \otimes \ket{v_n}$.
Given any tuple $\overline{v}$ and any value $v'$, a circuit $C$ is
said to \emph{approximate} $v'$ with probability $p \in [0,1]$ on
input $\overline{v}$, if $C\ket{\overline{v}}$ evaluates to $\ket \phi$, and
$\lvert \langle v' | \phi \rangle \rvert^2 \geq p$.

As discussed before, compilation of the QTRS will be done with respect to
an input structure.
Therefore, compilation yields a family of circuits indexed by the
structure of the possible inputs of $\mtf$, i.e., by the set
$\structset \mtf \triangleq \set{\struct{\overline{v}} \gmid
\overline{v} \in \funcin}$.
Note that the proof is constructive, thus actually generates the family
of circuits $\famcirc$.

\begin{restatable}[Circuit compilation]{theorem}{compile}\label{thm:compile}
  Let $\mfR \in \qtrscirc$ be a terminating QTRS of main symbol $\mtf$.
  Then, we can generate a family of circuits $\famcirc \triangleq (C_w)_{w \in
  \structset \mtf}$ s.t., for any input $\overline{v} \in \funcin$,
  if $\mtf(\overline{v}) \toqs v'$ then $C_{\struct{\overline{v}}}$
  approximates $v'$ with probability $\frac 2 3$ on input $\overline{v}$.
\end{restatable}

We are now interested in obtaining bounds on the size of the generated circuit.
Towards that end, an ordering $\ordfunceq$ on function symbols is introduced:
$\mtf \ordfunceq \mtg$ holds if $\mtf$ calls $\mtg$ in its reduction.
Formally, $\ordfunceq$ is defined as the smallest transitive relation
satisfying $\mtf \ordfunceq \mtg$ when there exists a context $C$ such that
$\mtf(\overline s) \to C[\mtg(\overline t)] \in \trsrule$.
We also define its induced strict order and equivalence relation respectively
by $\ordfunc$ and $\approx_\mfR$.
The \emph{rank} of a function symbol is defined inductively as follows,
with the convention that $\max(\emptyset) = 0$:
\[
  \rank \triangleq \max_{\mtf \ordfunc \mtg} (\rank[\mtg] + 1)
\]
\newcommand{\callee}{\mathrm{Callee}(\mtf, S)}
Function symbols are restricted by prohibiting mutal recursion
and enforcing that recursive calls are always performed on the same input.
This is defined formally below using the notation $\callee \triangleq
\{\mtf(\overline s) \gmid l \to C[\mtf(\overline s)] \in S\}$, for
$\mtf \in \mcF$ and for a $\mtf$-structural set $S$.

\begin{definition}
  A QTRS $\mfR$ is said to be \emph{simply recursive} if, for all function
  symbols $\mtf,\mtg$, $\mtf \approx_{\mfR} \mtg$ implies $\mtf = \mtg$ and,
  for any $\mtf$-structural set $S$, $\# \callee \leq 1$.
  Define $\qtrsrec$ as the set of simply recursive QTRS.
\end{definition}

These conditions are akin to~\cite{HPS25}, allowing to merge recursive
calls together and avoid an exponential blow-up in the size of the circuit;
they are verified for the QTRS of Example~\ref{ex:qft}.
We define the following syntactic sugar for a tuple
$\overline{v} =(v_1,\ldots,v_n)$:
$\size{\overline v} \triangleq \sum_{i=1}^n \size{v_i}$.

\begin{definition}
  Let $\tau : \mbN \to \mbN$ be a non-decreasing function.
  Define $\timeset$ as the set of QTRS of main symbol $\mtf$, where for any
  $\overline v \in \funcin$, $\mtf(\overline v)$ terminates in at most
  $\tau(\size{\overline v})$ steps.
\end{definition}

On the fragment of QTRS defined below, the size of the circuits
$\famcirc$, generated by the proof of Theorem~\ref{thm:compile}, can
be upper-bounded and related with the runtime-complexity of the QTRS:
\[
  \qtrstime \triangleq \qtrscirc \cap \qtrsrec \cap \timeset .
\]

\begin{restatable}{theorem}{compilebound}\label{thm:compile-bound}
  Take $\mfR \in \qtrstime$ of main symbol $\mtf$.
  Then, for any $\overline v \in \funcin$,
  $\size{C_{\struct{\overline v}}} = \mcO(\tau(\size{\overline v}
  )^{\rank + 1})$, where $C_{\struct{\overline v}} \in \famcirc$ is
  the compiled circuit of corresponding structure.
\end{restatable}

We end this section by providing an implicit characterization of the
complexity class $\fbqp$.
Recall that a family of circuits $(C_n)_{n \in \mbN}$ is said to be
\emph{uniform polynomially-sized} if there exists $P \in \mbN[X]$ such that
$\size{C_n} \leq P(n)$ and there is a polynomial-time Turing machine, which
takes $n$ as input and outputs a representation of $C_n$ for any $n \in \mbN$.

\begin{definition}[\cite{BV97}] \label{def:fbqp}
  A binary function $f : \set{0,1}^* \to \set{0,1}^*$ is said to be computed by
  a family of circuits $(C_n)_{n \in \mbN}$ if, for any $x \in \set{0,1}^*$,
  $C_{\size x}$ approximates $f(x)$ with probability $\frac 2 3$ on input $x$.
  $\fbqp$ is defined as the set of binary functions computed by a uniform
  polynomially-sized family of circuits.
\end{definition}

Define the fragment of QTRS terminating in polynomial time on lists of qubits.
\[
  \qtrspoly \triangleq \cup_{P \in \mbN[X]} \{
    \trsyntax_\mtf \in \qtrstime[P] \mid
    \sign \mtf = \typelist[\qbit] \to \typelist[\qbit]
  \}
\]

Using the characterization of $\fbqp$ from~\cite{Yam20} and
Theorem~\ref{thm:compile-bound}, $\qtrspoly$ is shown to correspond exactly to
$\fbqp$.
Let $\{\!\!\{ \qtrspoly \}\!\!\}$ be the set of binary functions that can be
computed by $\famcirc$ for $\mfR \in \qtrspoly$.

\begin{restatable}[$\fbqp$ characterization]{theorem}{fbqpcharac}
  \label{thm:fbqp}
  $\{\!\!\{ \qtrspoly \}\!\!\} = \fbqp$.
\end{restatable}

\section{Complexity Analysis}
\label{sec:complexity}
The previous section emphasized that termination of a QTRS is a crucial
hypothesis used for many results to hold, e.g., subject reduction
(Lemma~\ref{lem:subred}), type inference (Theorem~\ref{thm:type-inference}), or
circuit compilation (Theorem~\ref{thm:compile}).
In this section, we discuss how termination of a STRS can be proved, by
adapting existing techniques from standard TRS.
This is done by extending orders on classical terms to orders on any term with
superposition (\emph{worst path orderings}, Section~\ref{sec:wpo}).
We show how  these orders can be used to adapt well-known existing termination
techniques --- e.g., polynomial interpretations and
dependency pairs --- to the QTRS setting.
Finally, we also show how we can obtain complexity results for polynomial
time, yielding a decidable criterion to check whether a given QTRS
computes a function in $\mathtt{FBQP}$.

\subsection{Worst Path Orderings}\label{sec:wpo}

Recall first some definitions.
A \emph{quasi-order} $\ordeq$ is a reflexive and transitive relation.
A \emph{strict order} $\succ$ is an irreflexive and transitive relation.
For each quasi-order, we can associate the strict order $s \succ t \iff s
\ordeq t \wedge \neg(t \ordeq s)$.
If no infinite chain $t_1 \succ t_2 \dots$ can be built for a strict order,
it is called \emph{well-founded}.
A STRS $\mfR$ is said to be \emph{compatible} with a strict order $\succ$, if
for all rules $l \to r \in \trsrule$, $l \succ r$.
An order $\ord$ is said to be \emph{monotonic}, if $s \succ t$ implies
$\mtb(\dots, s, \dots) \ord \mtb(\dots, t, \dots)$ for all terms $s, t$ and for
all $\mtb \in \mcC \biguplus \mcF$.
An order $\ord$ is said to be \emph{closed under substitutions}, if
$s \succ t$ implies $s \sigma \ord t \sigma$, for all terms $s,
t$ and substitutions $\sigma$.
A \emph{rewrite order} is a monotonic order closed under substitutions.

On standard TRS, proofs of termination often include finding a
well-founded rewrite order $\ord$, such that if $s$ reduces to $t$, then $s
\ord t$.
However, existing orders are defined on classical terms.
To extend these orders to any term of a STRS, remark the following lemma.

\begin{restatable}[]{lemma}{chainsum}\label{lem:chain-sum}
  Let $\mfR$ be a STRS, and let $t = \sum_{i=1}^n \alpha_i \cdot t_i$ be a
  well-typed term. If $t$ starts an infinite chain on the order induced by
  $\toq$, then so does $t_i$ for some $1 \leq i \leq n$.
\end{restatable}

We also have the contrapositive result: if all $t_i$ terminate, then
so does $t$.
Therefore, the idea is to allow the order to always take the worst path, i.e.,
the term of a superposition taking the most time to reduce; this can also be
seen as taking the maximum over multiple paths.
This is formalized below.

\begingroup
\ebproofset{
  template=\small$\inserttext$,
}
\begin{definition}\label{def:wpo}
  Let $\ordeq$ be a quasi-order on classical terms. We define the \emph{worst
  path extension} of $\ordeq$ as the relation $\ordexteq$ defined
  inductively as follows:
  \[
    \begin{array}{c}
      \begin{prooftree}[]
        \hypo{s, t \in \trsetclassical\vphantom{\candef}}
        \hypo{s \ordeq t}
        \infer2{s \ordexteq t}
      \end{prooftree}
      \quad
      \begin{prooftree}
        \hypo{s \equiv \candef[s]}
        \hypo{\exists\,i,\,s_i \ordexteq t}
        \infer2{s \ordexteq t}
      \end{prooftree}
      \quad
      \begin{prooftree}
        \hypo{t \equiv \candefb t}
        \hypo{\forall\,j,\,s \ordexteq t_j}
        \infer2{s \ordexteq t}
      \end{prooftree}
    \end{array}
  \]
\end{definition}
\endgroup

This relation enjoys the following properties. In particular, it is actually a
(quasi-)order, hence the name worst path ordering.

\begin{restatable}{lemma}{wpoprop}\label{lem:wpo}
  Given a quasi-order $\ordexteq$,
  the following properties are satisfied,
  for any $s \equiv \candef[s]$ and $t \equiv \candefb t$:
  \begin{itemize}
    \item $s \ordexteq t \iff \forall j, \exists i, s_i \ordeq t_j$;
    \item If $s, t$ are classical, $s \ordexteq t \iff s \ordeq t$;
    \item $\ordexteq$ is a quasi-order.
    \item If $\ord$ is a rewrite order, then
      $\ordext$ is a rewrite order too.
  \end{itemize}
\end{restatable}

Furthermore, it can be used to obtain termination of STRS compatible
with the ordering.

\begin{restatable}{theorem}{wpotermination}\label{thm:wpo}
  Let $\ordeq$ be a quasi-order, where its induced strict order is a
  well-founded rewrite order.
  Then, any STRS compatible with the worst path extension $\ordext$ terminates.
\end{restatable}

Theorem~\ref{thm:wpo} thus allows us to generate a termination technique, for
each well-founded rewrite order $\ord$ on classical terms.
Such orderings have been deeply-studied in the literature and we now
show a few adaptations of some of these orderings to the quantum setting.

\subsection{Polynomial Interpretations}\label{sec:inter}

Polynomial interpretations were introduced in \cite{Lan79} as a way to show
termination of a TRS.
The main idea is to associate each symbol with a monotonic function
$\assign -$, called \emph{interpretation}, such that it decreases strictly over
any rule $l \to r$, i.e., $\assign l > \assign r$.
This is extended by defining interpretations over superpositions as the
maximum of the interpretations.

Equip $\mbN$ with the natural strict ordering $\ordnat$.
Given two polynomials $P, Q \in \polset n$,
define $P \ordnat Q$ as follows:
\[
  P \ordnat Q \iff \forall x_1, \dots, x_n \in \mbN, P(x_1, \dots, x_n) \ordnat
  Q(x_1, \dots, x_n)
\]

An \emph{assignment} $\assign -$ is a function mapping each symbol in $\mcC
\biguplus \mcF$ to a  multi-variate polynomial $\assign \mtb \in
\polset{\arity}$.
Define $\polvar$ as the countable set of indeterminates that can
be used in polynomials,
and fix a map $\rho : \mcX \to \polvar$ for the rest of the paper.
Assignments are extended to any term of a STRS inductively as follows:
\begin{align*}
  \assign \mtx &\triangleq \rho(\mtx) \\
  \assign{\mtb(s_1, \dots, s_n)} &\triangleq
  \assign{\mtb}(\assign{s_1}, \dots,
  \assign{s_n}) \\
  \assign{t} &\triangleq \max_{1 \leq i \leq n} \assign{t_i}, \qquad
  \text{provided }t \equiv {\sum}_{i=1}^n \alpha_i \cdot t_i \in \can
\end{align*}

where $\max_{1 \leq i \leq n} \assign{t_i}$ is a polynomial $P$
satisfying $P \ordnateq \assign{t_i}$.
By unicity of the canonical form (Lemma~\ref{lem:canonical-typed}),
the assignment of a term is properly defined.

Given an assignment $\assign -$, define the order $\ordassigneq$ as $s
\ordassigneq t \iff \assign s \ordnateq \assign t$.
An assignment $\assign -$ is \emph{monotonic} if $\ordassign$ is monotonic.
A STRS $\mfR$ is \emph{compatible} with an assignment $\assign -$,
if $\mfR$ is compatible with its ordering $\ordassign$.
As $\ordassigneq$ can be showed to satisfy the properties of a worst path
ordering, finding a compatible assignment provides a criterion for termination.

\begin{restatable}{theorem}{thminter}\label{thm:inter}
  Any STRS compatible with a monotonic assignment terminates.
\end{restatable}

\begin{example}
  Define the STRS $\mfR$ made with the following rules:
  \[
    \mcR = \left\{
      \begin{aligned}
        \mathtt{X}(\ket i) &\to \ket{1 - i} \\
        \mtf(\mtx, \nil) &\to \mtx :: \nil \\
        \mtf(\zket, h::t) &\to \mathtt{sgn}^+ \cdot \zket :: \mtf(h, t) +
        \mathtt{sgn}^+ \cdot \oket : \mtf(\mathtt{X}(h), t) \\
        \mtf(\oket, h::t) &\to \mathtt{sgn}^+ \cdot \zket :: \mtf(h, t) +
        \mathtt{sgn}^- \cdot \oket : \mtf(\mathtt{X}(h), t)
      \end{aligned}
    \right\}
  \]
  Then, given the following monotonic assignment:
  \[
    \assign{\ket i} = \assign \nil = 1 \qquad
    \assign{\consht}(X, Y) = X + Y + 1 \qquad
    \assign{\mathtt{X}} = X + 1\qquad
    \assign \mtf(X, Y) = X + 3Y
  \]
  Checking that $\mfR$ is compatible with this assignment is direct for the
  $\mathtt{X}$-rewrite rules and the first rule of $\mtf$; the check for the
  last two rules is done below:
  \[
    \assign{f(\ket i, h ::t)} =
    4 + 3\assign h + 3 \assign t \ordnat 3 + \assign h + 3 \assign t =
    \assign{\oket :: \mtf(\mathtt{X}(h), t)} \ordnat
    \assign{\oket :: \mtf(h, t)}
  \]
  Therefore, $\mfR$ terminates.
\end{example}

\subsection{Dependency Pairs}\label{sec:dpair}

Dependency pairs~\cite{AG00} is another termination technique, which has the
perk of dropping the strict monotonicity condition for rewrite rules.
The idea is to define, for any rewrite rule $l \to r$ and any subterm $s$ of
$r$ whose outermost symbol is a function symbol, a \emph{dependency pair}
$\dpair l s$.
Proving termination is equivalent to finding a quasi-ordering $\ordeq$
such that it decreases weakly on rewrite rules but strictly on
dependency pairs, i.e., $l \ordeq r$ and $l \ord s$.

This definition is adapted by requiring $t$ to be a subterm of the canonical
form of $r$, i.e., $r \equiv \candef[r]$ and $r_i = C[t]$ for some $i$, denoted
$C[t] \subcan r$ for conciseness.
As for worst path orderings, this definition will be sound, but not complete.

For each function symbol $\mtf \in \mcF$, we associate a corresponding
\emph{tuple symbol} $F$, of same signature.
For clarity, upper case symbols correspond to tuple symbols in the following.

\begin{definition}\label{def:dpair}
  Let $\mfR$ be a STRS. If $ \mtf(s_1, \dots, s_n) \to r \in \trsrule \wedge
  C[\mtg(t_1, \dots, t_m)] \subcan r, $ for $\mtg \in \trsfunc$, then
  $\dpair{F(s_1, \dots, s_n)}{G(t_1, \dots, t_m)}$ is called a dependency pair
  of $\mfR$.
\end{definition}

Chains are then defined as successive dependency pairs.

\begin{definition}\label{def:rchain}
  Let $\mfR$ be a STRS. A sequence of dependency pairs $\dpairdef 1,
  \dpairdef 2 \dots$ is a \emph{$\mfR$-chain}, if there exists a
  substitution $\sigma$, such that $t_j \sigma \toqs r_j$ and $s_{j+1}\sigma
  \subcan r_j$ for all $j$.
\end{definition}

The absence of infinite $\mfR$-chain implies termination, as in~\cite{AG00}.

\begin{restatable}[]{theorem}{rchain}\label{thm:rchain}
  A STRS $\mfR$ is terminating if no infinite $\mfR$-chain exists.
\end{restatable}

Equivalently, termination can be proven by finding orderings as defined below,
which implies the absence of infinite $\mfR$-chain.

\begin{restatable}[]{theorem}{dpairorder}\label{thm:dpair}
  A STRS $\mfR$ is terminating if there exists a well-founded weakly monotonic
  quasi-ordering $\ordeq$ on classical terms, where $\ord$ and $\ordeq$
  are closed under substitution, such that:
  \begin{itemize}
    \item $l \ordeq r_i$ for all rules $l \to \sumdef[r] \in \trsrule$ and for
      all $i$, with $\candef[r]$;
    \item $s \ord t$ for all dependency pairs $\dpair s t$.
  \end{itemize}
\end{restatable}

\subsection{Complexity Results}

Sections~\ref{sec:inter} and~\ref{sec:dpair} have provided different
techniques, based on
Lemma~\ref{lem:chain-sum}, for proving the termination of a STRS.
We now are interested in obtaining complexity results on the runtime-complexity,
which rely on the following lemma.

\begin{restatable}{lemma}{redsum}\label{lem:red-sum}
  Let $\mfR$ be a STRS, and let $t = \sumdef \in \trsetground$.
  Assume each $t_i$ terminates in at most $k_i$ steps.
  Then $t$ terminates in at most $\max_{1 \leq i \leq n} k_i$ steps.
\end{restatable}

As a consequence of Lemma~\ref{lem:red-sum}, the assignment of a term bounds
its runtime complexity.

\begin{restatable}{lemma}{interbound}\label{lem:inter-bound}
  Let $\mfR$ be a QTRS compatible with an assignment $\assign -$.
  Then, given any term $t$, it terminates in at most $\assign t$ steps.
\end{restatable}

Despite their name, polynomial interpretations are not necessarily
bounded by polynomials. This can, however, be obtained by
imposing a structure for constructor assignments.

\begin{definition}\label{def:additive}
  Let $\mfR$ be a QTRS, and let $\assign -$ be an assignment.
  We say that it is an \emph{additive assignment}, if for all $\mtc
  \in \trscons$ of arity $n > 0$, $\assign \mtc = \sum_{i=1}^n X_i +
  \alpha_\mtc$, with $\alpha_\mtc \geq 1$.
\end{definition}

Such a restriction allows bounding any assignment by a polynomial,
and thus to obtain polynomial time.

\begin{restatable}{theorem}{polytime}\label{thm:polytime}
  Let $\mfR$ be a QTRS compatible with an additive assignment
  and let $\mtf \in \mcF$.
  Then, there is a polynomial $P$ such that for any
  $\overline v \in \funcin$,
  $\mtf(\overline v)$ terminates in time $P(\size v)$.
\end{restatable}

As a consequence of Theorem~\ref{thm:fbqp}, this yields a decidable way of
checking whether a compilable QTRS computes a binary function in the
quantum complexity class $\fbqp$.

\bibliography{refs}

\appendix
\section{Additional Material}\label{app:additional}

\begingroup
\newcommand{\sgn}{\mathtt{sgn}}
\begin{example}
  Let us come back on the STRS defined in Example~\ref{ex:had-gen},
  and prove that it is a QTRS.
  The function symbol $\mtf$ comes with the signature
  $\sign \mtf = \nat \to \qbit \to \qbit$.
  First, prove the typing condition; and suppose that
  $\typing \Gamma \Delta {\mtf(n, \zket)} T$.
  By Lemma~\ref{lem:tree-canonical}, as the term terminates,
  its typing tree contains only classical rules.
  The only applyable rule is the symbol rule. Going upwards
  gives the following typing tree:
  \[
    \begin{prooftree}
      \infer0{\typing{\Gamma, n : \nat}{}{n}{\nat}}
      \infer0{\typing{\Gamma, n : \nat} {} \zket \qbit}
      \infer2{\typing{\Gamma, n : \nat} {}{\mtf(n, \zket)}  \qbit}
    \end{prooftree}
  \]
  Thus, imposing $T = \qbit$ and $\Delta = \emptyset$.
  The right term can be typed identically as follows,
  \[
    \begin{prooftree}[template = \tiny$\inserttext$, separation = 0.8em]
      \infer0{\typingamp{\Gamma, n : \nat}{\sgn^+}}
      \infer0{\typing{\Gamma, n : \nat}{}{n}{\nat}}
      \infer1{\typingamp{\Gamma, n : \nat}{\mta^+(n)}}
      \infer0{\typing{\Gamma, n : \nat} {}{\ket i} \qbit}
      \hypo{
        \zket \orthored \oket \\
        \forall \sigma \in \subcont[{\Gamma, n : \nat}],
        \size{\inter{\sgn^+}}^2 + \size{\inter{\mta^+(\underline \iota)}}^2 = 1
      }
      \infer4{\typing{\Gamma, n : \nat} {}{\mtf(n, \zket)}  \qbit}
    \end{prooftree}
  \]
  The amplitude condition can be verified straightforwardly, for any
  substitution of $n$ by a ground natural number $\underline \iota$.
  This process can be done for the other rewrite rule, and validates
  the second condition of Definition~\ref{def:qtrs}.
  It now remains to show that $\mtf$ is an isometry.
  Consider two inputs $v, w \in \funcin$.
  By typing, they are of the shape $v = (\alpha \cdot \zket + \beta \cdot
  \oket, \underline \iota)$ and $w = (\gamma \cdot \zket + \delta \cdot \oket,
  \underline{\iota'})$.
  if $v \orthored w$, this implies the following:
  \[
    \begin{aligned}
      S &= \inter{\alpha}\inter{\gamma}^*
      \kronweird{(\zket, \underline \iota)}{(\zket, \underline{\iota'})} +
      \inter{\alpha}\inter{\delta}^*
      \kronweird{(\zket, \underline \iota)}{(\oket, \underline{\iota'})} \\
      + &\inter{\beta}\inter{\gamma}^*
      \kronweird{(\oket, \underline \iota)}{(\zket, \underline{\iota'})} +
      \inter{\beta}\inter{\delta}^*
      \kronweird{(\oket, \underline \iota)}{(\oket, \underline{\iota'})}
      = 0
    \end{aligned}
  \]
  For $\delta_\perp$ to be defined, it imposes $\underline \iota =
  \underline{\iota'}$; and the equality implies $\inter{\alpha}\inter{\gamma}^*
  + \inter{\beta}\inter{\delta}^* = 0$.
  Now, it remains to prove that $\mtf(\alpha \cdot \zket + \beta \cdot \oket,
  \underline \iota) \orthored \mtf(\alpha \cdot \zket + \beta \cdot \oket,
  \underline \iota)$.
  Developping each term through $\equiv$, and then reducing them via (Q)
  yields the two reduced:
  \[
    (\alpha \cdot \sgn^+ + \beta \cdot \mta^-(\underline \iota)) \cdot \zket +
    (\alpha \cdot \mta^+(\underline \iota) + \beta \cdot \sgn^-) \cdot \oket
  \]
  and
  \[
    (\gamma \cdot \sgn^+ + \delta \cdot \mta^-(\underline \iota)) \cdot \zket +
    (\gamma \cdot \mta^+(\underline \iota) + \delta \cdot \sgn^-) \cdot \oket
  \]
  These terms are also expressed as canonical forms.
  Again, expressing the orthogonality yields the following,
  as $\kronweird{\ket i}{\ket{1 - i}} = 0$:
  \[
    \begin{aligned}
      S &= \inter{\alpha \cdot \sgn^+ + \beta \cdot \mta^-(\underline \iota)}
      \inter{\gamma \cdot \sgn^+ + \delta \cdot \mta^-(\underline \iota)}^*
      + \inter{\alpha \cdot \mta^+(\underline \iota) + \beta \cdot \sgn^-}
      \inter{\gamma \cdot \mta^+(\underline \iota) + \delta \cdot \sgn^-}^*
    \end{aligned}
  \]
  By computing the amplitudes, we obtain
  $S =\inter{\alpha}\inter{\gamma}^* + \inter{\beta}\inter{\delta}^* = 0$.
  Therefore, $\mtf$ is an isometry, and $\mfR$ is a QTRS.
\end{example}
\endgroup

\section{Proofs of Section~\ref{sec:results}}

\subsection{Proofs of Section~\ref{sec:properties}}

We first introduce an operator $\qt s t$, called the \emph{quantity}, which
aims to measure the amplitude of a classical term s in a term $t$.
This will be in particular useful to show Lemma~\ref{lem:canonical-unicity}.
Suppose $\mcA$ contains two symbols $\mta_0, \mta_1$, where
$\inter{\mta_0} = 0$ and $\inter{\mta_1} = 1$.
Given a term $t$, denote $\subcont[\fv t]$ as the set of substitutions
such that $t \sigma$ is a ground term.

\begin{definition}\label{def:quantity}
  Let $s$ be a classical term.
  We define the following map $\theta_s$ from terms to amplitudes as follows:
  \begingroup
  \renewcommand{\arraystretch}{1.5}
  \[
    \begin{aligned}
      \qt s \mtx &\triangleq \kron s \mtx \\
      \qt s {\mtb(t_1, \dots, t_n)} &\triangleq
      \begin{cases}
        \prod_{i=1}^n \qt{s_i}{t_i} \text{ if } s = \mtb(s_1, \dots, t_n) \\
        \mta_0 \text{ else }
      \end{cases} \\
      \qt s {t_1 + t_2} &\triangleq \qt s {t_1} + \qt s {t_2} \\
      \qt s {\alpha \cdot t} &\triangleq \alpha \cdot \qt s t
    \end{aligned}
  \]
  \endgroup
  where the notation $\kron s t$ outputs $\mta_1$ if $s = t$, and $\mta_0$ else.
\end{definition}

\begin{lemma}\label{lem:quantity-one}
  Let $s$ be a classical term.
  Then $\qt s s = \mta_1$.
\end{lemma}

\begin{proof}
  Proven by induction on $s$, direct by definition of the quantity. \qedhere
\end{proof}

\begin{lemma}\label{lem:quantity-equiv}
  Let $s$ be a classical term and $t, t'$ be two terms.
  If $t \equiv t'$, then for any substitution $\sigma \in \subcont[\fv t] \cap
  \subcont[\fv{t'}]$, $\inter{\qt s t \sigma} = \inter{\qt s {t'} \sigma}$.
\end{lemma}

\begin{proof}
  By induction on each rule defining $\equiv$ in Figure~\ref{tab:equiv};
  the transitive and reflexive case can be derived from this.
  The first two lines of rules are verified directly as $\mbC$
  is a vector space;
  while the linearity rules follow directly the case of summation
  in Definition~\ref{def:quantity}.
  The last rule with the equivalence context can be proven by induction, on the
  construction of the context, the base case being verified above. \qedhere
\end{proof}

\begin{lemma}\label{lem:quantity-classical}
  Let $s, t$ be two classical terms.
  Then $\qt s t = \kron s t$.
\end{lemma}

\begin{proof}
  Direct by induction on $t$; as it is classical, only the first two
  cases of Definition~\ref{def:quantity} are considered. \qedhere
\end{proof}

As a consequence, if $t = \sumdef$ with $t_i$ being classical terms,
$\qt s t = \sum_{i=1}^n \alpha_i \cdot \kron s {t_i}$.

\begin{lemma}[Unicity]\label{lem:canonical-unicity}
  If $t_0$ has a canonical form, then it is unique, up to reordering
  and amplitude equivalence.
\end{lemma}

\begin{proof}
  Suppose $t_0$ has two canonical forms $t = \sum_{i=1}^n \alpha_i
  \cdot t_i$ and $s = \sum_{j=1}^m \beta_j \cdot s_j$, thus $s \equiv t$.
  Take any substitution $\sigma \in \subcont[\fv s] \cap \subcont[\fv
  t]$, and any
  $1 \leq i_0 \leq n$.
  By definition of the quantity and by Lemma~\ref{lem:quantity-classical},
  $\inter{\qt{t_{i_0}} t \sigma} = \inter{\alpha_{i_0} \sigma}$, which by
  definition of a canonical form, is different from $0$.
  By Lemma~\ref{lem:quantity-equiv}, $\inter{\qt{t_{i_0}} t \sigma} =
  \inter{\qt{t_{i_0}} s \sigma} = \sum_{j=1}^m \inter{\beta_j} \cdot \kron
  {t_{i_0}}{s_j}$.
  If $\kron {t_{i_0}}{s_j} = \mta_0$ for all $j$, then $\inter{\qt{t_{i_0}} s
  \sigma} = 0$, which will contradict Lemma~\ref{lem:quantity-equiv}; and if
  there are two $j$ where it is $\mta_1$, then we will have $s_j = t_{i_0} =
  s_j'$ which contradicts the fact that $s$ is a canonical form.
  Therefore, there exists $j_0$ such that $t_{i_0} = s_{j_0}$,
  and $\alpha_{i_0}$ is equivalent to $\beta_{j_0}$.
  Thus, any $t_i$ has a unique match in $(s_j)$, therefore $n \leq m$.
  The same process can be done the other way around, so $m \leq n$,
  thus $m = n$.
  So each sum contains the same number of elements, such that there
  is a map from $1, \dots, n$ to $1, \dots, n$ that associates $i$
  with $j$ such that $t_i = s_j$ and $\alpha_i = \beta_j$.
  Therefore, both sums are equal up to permutation of the $t_i$.
\end{proof}

\begin{restatable}[]{lemma}{canonical}\label{lem:canonical}
  Let $\mfR$ be a STRS, and $t$ be a term. Then, either $\equivz$,
  or $t \equiv \candef$, and this canonical
  form is unique, up to sum reordering and amplitude equivalence.
\end{restatable}

\begin{proof}
  Unicity being direct from Lemma~\ref{lem:canonical-unicity}, let us prove
  either the existence of a canonical form or $\equivz$, by
  induction on the syntax of $t$.
  Note that if $\equivz$, then $t$ has no canonical form.
  \begin{itemize}
    \item If $t = \mtx$, then $t$ has $\mta_1 \cdot \mtx$ as a canonical form,
      which is equivalent to $t$ directly by $\equiv$.
    \item Suppose $t = \mtb(t_1, \dots, t_n)$. By induction, either $t_i \equiv
      \mta_0 \cdot t_i$ is verified for some $i$, in which case $t \equiv
      \mta_0 \cdot \mtb(t_1, \dots, t_n)$; else, all $t_i$ possess a canonical
      form.
      Therefore, we have $t_i \equiv \sum_{j_i = 1}^{m_i} \alpha_{j_i} \cdot
      t_{j_i}$.
      By developing $t$ under $\equiv$, we have:
      \[
        t \equiv \sum_{j_1=1}^{m_1} \dots \sum_{j_n=1}^{m_n} \alpha_{j_1} \dots
        \alpha_{j_n} \cdot \mtb(t_{j_1}, \dots, t_{j_n})
      \]
      This can be seen as the following sum, where $\alpha_k =
      \alpha_{j_1} \dots \alpha_{j_n}$:
      \[
        t \equiv \sum_{k=1}^{m_1 \times \dots \times m_n} \alpha_k
        \cdot \mtb(t_1^k, \dots, t_n^k)
      \]
      In particular, $\inter{\alpha_k} \neq 0$ as $\inter{\alpha_{j_i}} \neq 0$;
      and as, for fixed $i$, all $t_{j_i}$ are pairwise distinct, so are the
      $\mtb(t_1^k, \dots, t_n^k)$. Therefore, this is a canonical form.
    \item Suppose $t = \sumdef$.
      By induction hypothesis, either $t_i \equiv \mta_0 \cdot t_i$, in which
      case it can be removed through $\equiv$, or it possesses a canonical
      form.
      One can thus rewrite $t \equiv \sum_{j=1}^m \alpha_j' \cdot t_j'$,
      where $t_j'$ are the $t_i$ such that $\nequivz[t_i']$.
      Let us denote $p_k$ as the set of all classical terms present in
      the canonical form of at least one $t_j'$. Up to adding them with
      $\equiv$ ($t \equiv t + \mta_0 \cdot p_k$), we can write each $t_j'$ as
      $\sum_{k=1}^l \beta_{jk} \cdot p_k$, and thus write $t$ as follows:
      \[
        t \equiv \sum_{j=1}^m \alpha_j' \cdot (\sum_{k=1}^l
        \beta_{jk} \cdot p_k)
        \equiv \sum_{k=1}^l \gamma_k \cdot p_k
      \]
      with $\gamma_k = \sum_{j=1}^m \alpha_j' \beta_{jk}$.
      By definition, all $p_k$ are classical and pairwise distinct.
      Now, either $\inter{\gamma_k} = 0$ for all $k$, and $t \equiv \mta_0 \cdot
      t'$; or $\inter{\gamma_k} \neq 0$ holds for some $k$, and the obtained
      term is a canonical form, up to removing the other terms. \qedhere
  \end{itemize}
\end{proof}


\begin{restatable}[Weakening]{lemma}{weakening}\label{lem:weakening}
  Let $\mfR$ be a QTRS, and let $\typingdef$ be a well-typed term.
  Then, for any $\Gamma' \supseteq \Gamma$ such that $\Gamma' \cap \Delta =
  \varnothing$, $\typing{\Gamma'} \Delta t T$.
\end{restatable}

\begin{proof}
  Direct by induction on typing.
\end{proof}


\begin{lemma}\label{lem:confluence-equiv}
  Let $\mfR$ be a STRS, and let $s, t$ be two terms.
  Suppose $s \equiv t$, $s \toq s'$ and $t \toq t'$.
  Then, $t \equiv t'$.
\end{lemma}

\begin{proof}
  This is proven by induction on the size of the derivation of both
  $\toq$ relations.

  If either $s$ or $t$ reduces by the rule (E), assuming $s$ without
  loss of generality, then $s \equiv s_1$, $s_1 \toq s_1'$ and $s_1' \equiv s'$.
  Thus, $t \equiv s_1$, by induction hypothesis, $s_1' \equiv t'$, and
  by transitivity of $\equiv$, $s' \equiv t'$.
  We now suppose that both $s$ and $t$ do not reduce via (E).

  If $s, t$ both reduce via the rule (C), then they are both classical.
  By quantity, $1 = \inter{\qt s s} = \inter{\qt s t} = \inter{\kron s t}$,
  thus $s = t$, with $s = E[l \sigma]$ and $t = E'[l'\sigma']$.
  As $E$ reduces from left to right and $s = t$ then $E = E'$; and as $\trsrule$
  is orthogonal, there is a unique choice for the rewrite rule,
  thus $l = l'$, and $\sigma = \sigma'$.
  Therefore, by the rule (C), $s' = t'$.

  Suppose $s, t$ both reduce via (Q), thus $t = \sum_{i=1}^n \alpha_i \cdot
  t_i$, and $s = \sum_{j=1}^m \beta_j \cdot s_j$ (same for $s', t'$ by adding
  $'$ on $t_i$ and $s_j$).
  By Lemma~\ref{lem:canonical-unicity}, their canonical form are equal, up to
  reordering and amplitude equivalence; in particular, they have the same number
  of elements.
  Therefore, for any $t_i$, there exists $s_j$ such that $t_i = s_j$; either
  they are normal forms, in which case $t_i' = s_j'$; or they both reduce, in
  which case by induction hypothesis $t_i' \equiv s_j'$.
  Furthermore, for these terms, $\alpha_i$ is equivalent to $\beta_i$; thus
  one can rewrite $t' \equiv t' + (\beta_j + \mta_{min} \cdot \alpha_i) \cdot
  t_i$, with $\inter{\mta_{min}} = -1$, as $\inter{\beta_j + \mta_{min} \cdot
  \alpha_i} = 0$; this thus transforms $\alpha_i \cdot t_i$ in $\beta_j \cdot
  t_i$.
  Therefore, starting from $t'$, one can reorder all terms as in $s'$, replace
  $\alpha_i$ by $\beta_j$, and replace $t_i$ by $s_j$; all of this can be done
  via $\equiv$, and we reach $s'$. Therefore, $t' \equiv s'$.

  Finally, suppose $s$ reduces via (C) and $t$ via (Q), thus $t =
  \sumdef$.
  By unicity of the canonical form, as $\mta_1 \cdot s$ is a canonical form,
  $t = \alpha \cdot t_1$, where $t' = s$ and $\inter \alpha =1$.
  By definition, $t_1$ must reduce to some $t_1'$;
  by induction hypothesis, $s' \equiv t_1'$.
  As $\inter \alpha = 1$, $t' = \alpha \cdot t_1' \equiv t_1' = s'$,
  thus we conclude.

\end{proof}

\confluence*

\begin{proof}
  The fact that the canonical form is unique is direct by
  Lemma~\ref{lem:canonical-unicity}.
  Let us prove semi-confluence, which implies confluence.
  Let $t$ be a term, suppose $t \toq t_1$ and $t \toqs t_2$,
  and prove that there exist $t_3, t_4$ such that $t_1 \toqs t_3$,
  $t_2 \toqs t_4$, and $t_3 \equiv t_4$.
  This is done by induction on the number of rewrite steps of $t \toqs t_2$.
  If it is done in $0$ steps, then take $t_4 = t_3 = t_1$.
  If it is done in one step, then by Lemma~\ref{lem:confluence-equiv},
  $t_1 \equiv t_2$, thus take $t_3 = t_1$ and $t_4 = t_2$.
  If it is done in two steps or more, i.e.,
  $t \toq t' \toq t'' \toqs t_2$, by Lemma~\ref{lem:confluence-equiv},
  $t' \equiv t_1$; via (E) rule of Figure~\ref{tab:toq},
  $t_1 \toq t''$. Therefore, take $t_3 = t_4 = t_2$.
\end{proof}

\begin{lemma}\label{lem:red-nequivz}
  Let $\mfR$ be a QTRS, and let $t$ be a term.
  Suppose $t \toq t_1$.
  Then, $\nequivz$.
\end{lemma}

\begin{proof}
  Direct by induction on $\toq$; in the case where $t$ is classical,
  conclude by Lemma~\ref{lem:quantity-equiv} by checking
  $\qt t t$ and $\qt t {\mta_0 \cdot t}$.
\end{proof}

\begin{lemma}\label{lem:red-can}
  Let $\mfR$ be a QTRS, and let $s$ be a term. Suppose $s \toq t$.
  Then $s \equiv \candef[s]$, $t \equiv \sumdef[t]$, and $s_i \toqif t_i$.
\end{lemma}

\begin{proof}
  First, Lemma~\ref{lem:red-nequivz} states that $\nequivz[s]$,
  and thus Lemma~\ref{lem:canonical} implies that $s$ has a canonical
  form $\candef[s]$.
  Now, by induction on the derivation of $s \toq t$:
  \begin{itemize}
    \item If $s = C[\sigma l]$, as it is a classical term,
      $s \equiv \mta_1 \cdot s \in \can$, and $t \equiv \mta_1 \cdot t$,
      thus we conclude directly;
    \item If $s$ is equivalent to some $s'$, $s'$ has the same canonical
      form as $s$, thus by induction hypothesis we can conclude.
    \item If $s$ is directly a canonical form, we conclude. \qedhere
  \end{itemize}
\end{proof}

\confluence*


To ease the notations, the rules of Figure~\ref{tab:typing} typing
term equivalence, a symbol $\mtb$, and a superposition, are respectively
denoted with labels (equiv), (symb), and (sup) in the following proofs.

A typing derivation is said to be \emph{classical}, if it contains
neither (sup) nor (equiv) typing rules.

\begin{lemma}\label{lem:tree-classical}
  Let $\mfR$ be a QTRS, and let $\typingdef$ be a well-typed classical term.
  Then, it can be typed by a classical typing derivation.
\end{lemma}

\begin{proof}
  First show that for any $\typingdef$ well-typed term, and any classical
  term $s$, if $\inter{\qt s t} \neq 0$, then $\typingdef[s]$ holds with a
  classical typing derivation.
  This is proven by induction on the typing of $t$.
  \begin{itemize}
    \item If $t$ is typed as a variable, then it is classical, thus
      $\inter{\qt s \mtx} \neq 0$ implies that $s = \mtx$, thus we conclude.
    \item If $t = \mtb(t_1, \dots, t_n)$, then quantity implies that
      $s = \mtb(s_1, \dots, s_n)$ with $\inter{\qt{s_i}{t_i}} \neq 0$.
      By induction hypothesis, we obtain typing of $s_i$ with same type
      and context as $t_i$ via a classical derivation, thus we conclude.
    \item If $t \equiv t'$, Lemma~\ref{lem:quantity-equiv} implies that
      $\inter{\qt s {t'}} \neq 0$, thus we apply the induction hypothesis and
      conclude, as $t'$ and $t$ have the same type and context.
    \item Finally, suppose $t = \sumdef$. By quantity, if $\inter{\qt s {t_i}}=
      0$ for all $i$, then $\inter{\qt s t} \neq 0$ fails. Thus, apply the
      induction hypothesis on $\qt s {t_i}$ where $\inter{\qt s {t_i}} \neq 0$
      holds, and conclude as $t_i$ has same type and context as $t$.
  \end{itemize}
  We conclude by applying this result for $t = s$, as $\qt s s = \mta_1$ by
  Lemma~\ref{lem:quantity-one}.
\end{proof}

\begin{lemma}\label{lem:tree-canonical}
  Let $\mfR$ be a QTRS, and let $\typingdef$ be a well-typed terminating term.
  Then, it can be typed as follows:
  \[
    \begin{prooftree}
      \infer0{\typing{\Gamma}{\Delta}{s_0} T}
      \infer1[(sup)]{\typing{\Gamma}{\Delta} {s_1} T}
      \ellipsis{}{}
      \infer1[(sup)]{\typing{\Gamma}{\Delta}{s_n} T}
      \hypo{s_n \equiv t}
      \infer2[(equiv)]{\typingdef}
    \end{prooftree}
  \]
  where $t$ types $s_0$ via a classical typing derivation, and then is followed
  by $n \geq 0$ (sup) typing rules, and one (equiv) typing rule.
\end{lemma}

\begin{proof}
  The derivation is built by modifying the original typing derivation of $t$,
  with the following remarks.
  Note that below, when we say that a typing rule (a) commutes with (b), we
  only mean that a derivation with root (b) and direct premise (a) can be
  rewritten as a derivation with a root (a) taking as direct premise (b).
  \begin{itemize}
    \item First, an (equiv) typing rule commutes with any other rule, thanks to
      the last rule of Figure~\ref{tab:equiv}; and by transitivity of $\equiv$,
      two (equiv) rules can be merged together.
      Any typing derivation for the following case can be considered
      with only one (equiv) rule at the end.
    \item We also need to prove that (sup) commutes with (symb) for terminating
      terms, up to adding an (equiv) rule.
      Assuming that (sup) happens in the first slot without loss of generality,
      if $\mtb(\sumdef, \dots, s_m)$ is typable through (symb) and (sup), we
      want to prove that $\sum_{i=1}^n \alpha_i \cdot \mtb(t_i, \dots, s_m)$ is
      typable through (sup) and (symb), thus the initial term is typable
      through (equiv), (sup) and (symb).
      The contexts and types will be correct for the typing to happen, and the
      phases are normalized, but the difficult part is to show that if $t_i
      \orthored t_j$, then $\mtb(t_i, \dots, s_m) \orthored \mtb(t_j, \dots,
      s_m)$.
      However, if $\mtb$ is a constructor, then the terms will reduce to
      $\mtb(v_i, \dots, w_m)$, where $v_i \orthored v_j$ as $t_i \orthored
      t_j$, thus orthogonality is recovered;
      and if $\mtb$ is a function symbol, it must be an isometry by definition
      of a QTRS, and thus orthogonality is also preserved (up to reducing the
        $t_i$ to their normal forms $v_i$, which can be done as they are
      orthogonal, thus terminate).
      Furthermore, considered terms always terminate. Indeed, either it is part
      of a construct from the original terminating term $t$; else, it must have
      been introduced through $\equiv$, but thus introduced a summation,
      which, if typed, implies that it terminates.
  \end{itemize}
  Using these remarks, one is able to commute the rules accordingly to obtain
  the expected structure.
\end{proof}

\begin{lemma}\label{lem:sub-lemma-one}
  Let $\mfR$ be a QTRS, and let $\typingclosed s T$ be a well-typed classical
  term.
  Suppose $s = t \sigma$, for $t$ a left-linear classical term.
  Then, there exist $\Gamma, \Delta$ such that $\typing \Gamma \Delta t T$, and
  $\sigma \in \subcont[\Gamma \cup \Delta]$.
\end{lemma}

\begin{proof}
  Recall that $\sigma$ maps variables to classical values.
  By induction on the syntax of $t$, either $t = \mtx$, in which case $\typing
  \Gamma \Delta \mtx T$ is valid through one of the two variable typing rules,
  taking $\typedom{\Gamma} \cup \typedom{\Delta} = \set{\mtx}$;
  and $\sigma$ maps $\mtx$ to $s$, which is a classical and closed value of
  good type.
  Else, $t = \mtb(t_1, \dots, t_n)$, which implies that $s = \mtb(s_1, \dots,
  s_n)$ by action of a substitution, with $s_i = t_i \sigma$.
  As $s$ is classical, Lemma~\ref{lem:tree-classical} states that it has a
  classical typing derivation, thus we have $\signdef$, and $\typing \Gamma
  {\Delta_i} {t_i} T_i$.
  By induction on each $t_i$, $\typing \Gamma {\Delta_i}{s_i} T_i$,
  and $\sigma \in \subcont[\Gamma \cup \Delta_i]$.
  Typing with (symb) yields the expected typing result; and as $t$ is
  left-linear, variables in the $t_i$ are disjoint, and thus we can conciliate
  all possible cases to obtain $\sigma \in \subcont[\Gamma \cup \Delta]$, and
  conclude.
\end{proof}

\begin{lemma}\label{lem:sub-lemma-two}
  Let $\mfR$ be a QTRS, let $\typing \Gamma \Delta t T$ be a
  well-typed term, and let $\typingamp \Gamma \alpha$
  be a well-typed amplitude.
  Then, for any substitution $\sigma \in \subcont[\Gamma \cup \Delta]$,
  $\typingclosed{t \sigma} T$ and $\typingamp{}{\alpha \sigma}$.
\end{lemma}

\begin{proof}
  By induction on typing of both amplitudes and terms.
  \begin{itemize}
    \item If $t = \mtx$, then $\mtx : T \in \Gamma \cup \Delta$, thus
      $\typingclosed{s\sigma} T$ is direct by definition of a
      context substitution.
    \item Suppose $t \equiv s$, with $\typing \Gamma \Delta s T$.
      By induction hypothesis, $\typingclosed{s \sigma} T$,
      and it is easy to verify that $s \sigma \equiv t \sigma$, as
      $\sigma$ maps to classical values, thus the structure stays
      identical. Therefore, $\typingclosed{t \sigma} T$.
    \item If $t = \mtb(t_1, \dots, t_n)$, typing implies that $\Delta =
      \Delta_1, \dots, \Delta_n$, with $\typing \Gamma {\Delta_i}{t_i}{T_i}$.
      Now, $\sigma$ is, a fortiori, a $\Gamma \cup \Delta_i$-context
      substitution.
      Therefore, by induction hypothesis, $\typingclosed{t_i \sigma}{T_i}$,
      thus $\typingclosed{\mtb(t_1\sigma, \dots, t_n \sigma)}{T}$, and conclude
      as $\mtb(t_1\sigma, \dots, t_n \sigma) = t \sigma$.
    \item Suppose $t = \sumdef$.
      By typing, $\typing \Gamma \Delta {t_i} T$, thus $\typingclosed{t_i
      \sigma} T$ by induction hypothesis.
      By typing, $\typingamp \Gamma \alpha_i$, thus $\typingamp{} \alpha_i
      \sigma$ by induction hypothesis;
      By typing, $\sum_{i=1}^n \size{\inter{\alpha_i \delta}}^2 = 1$ is true
      for any substitution $\delta \in \subcont$, thus as $\sigma$ is, a
      fortiori, a $\Gamma$-context substitution, it is true for $\sigma$, and
      now $\alpha_i \sigma$ are closed, thus no need to verify it for another
      substitution.
      Finally, $t_i \perp t_j$ implies that for any substitution $\delta \in
      \subcont[\Gamma \cup \Delta]$, $t_i$, $t_j$ reduce to some value which
      verifies some properties. Again, this is true in particular for $\sigma$,
      and now $t_i \sigma$ is closed, thus orthogonality is verified between
      $t_i \sigma$ and $t_j \sigma$.
      Therefore, one is able to type $\typingclosed{\sum_{i=1}^n (\alpha_i
      \sigma) \cdot t_i \sigma}{T}$, which, by definition of $\sigma$, is equal
      to $t\sigma$.
    \item If $\alpha = \mta(\alpha_1^c, \dots, \alpha_n^c)$, this is direct
      by induction as for $\mtb(t_1, \dots, t_n)$.
    \item The case for $\alpha_1 + \alpha_2$ and $\alpha_1 \cdot \alpha_2$ are
      also direct by induction and by definition of a substitution. \qedhere
  \end{itemize}
\end{proof}

\subred*

\begin{proof}
  By induction on $\toq$.
  \begin{itemize}
    \item Suppose $s = C[l \sigma]$ and $t = C[r \sigma]$.
      As $s$ is classical, Lemma~\ref{lem:tree-classical} states that
      it can be typed by a classical typing derivation, which thus will
      coincide with each symbol of $C$.
      Therefore, $\typingclosed{l\sigma}{T'}$; as QTRS are left-linear, so is
      $l$.
      By Lemma~\ref{lem:sub-lemma-one}, there exist $\Gamma, \Delta$ such that
      $\typing \Gamma \Delta l T'$, and $\sigma \in \subcont[\Gamma \cup
      \Delta]$.
      As $\mfR$ is a QTRS, $\typing \Gamma \Delta r T'$.
      By Lemma~\ref{lem:sub-lemma-two}, $\typingclosed {r\sigma} T'$;
      and one can type $t$ by the same derivation as $s$, by replacing
      $l \sigma$ by $r \sigma$, as it has the same context and type,
      and it is a classical derivation, thus no need to derive orthogonality.
    \item Suppose that $s \equiv s'$ and $t \equiv t'$ with $s' \toq t'$.
      As $\equiv$ preserves typing, $\typing \Gamma \Delta {s'} T$;
      by induction hypothesis, $\typing \Gamma \Delta {t'} T$;
      and by equivalence again, $\typingdef$.
    \item Finally, suppose that $s = \candef[s]$ and $t = \sumdef$,
      with $s_i \toqif t_i$ for all $i$.
      As stated in Lemma~\ref{lem:tree-canonical}, consider a typing derivation
      where all equivalence relations are at the end, thus $s \equiv s'$, where
      $s'$ is well-typed, with no equivalence relation inside its derivation.
      If $s'$ is classical and $s$ is a canonical form, then $s = 1
      \cdot s'$ by quantity;
      by induction hypothesis, this implies $\typing \Gamma \Delta {t'} T$,
      and thus $\typing \Gamma \Delta t T$  by $\equiv$, with $t = 1 \cdot t'$.
      Now, suppose $s'$ is not classical, and thus $s'$ terminates.
      Lemma~\ref{lem:tree-canonical} implies the following, with $s'
      \equiv s''$:
      \[
        s'' = \sum_{i_1=1}^{m_1} \alpha_{i_1} \cdot (\dots
        \sum_{i_n=1}^{m_n} \beta_{i_1 \dots i_n} \cdot p_{i_1 \dots i_n}),
      \]
      where the $p_{ij}$ are classical terms.
      Developing all the sums yields a canonical form, up to removing potential
      terms with a phase equivalent to $0$; and this is exactly $s$, by unicity
      of the canonical form (Lemma~\ref{lem:canonical-unicity}).
      By each (sup) rule, we also have $\typing \Gamma \Delta {p_{i_1 \dots
      i_n}} T$, and such term is equal to some $s_i$.
      As $s_i \toqif t_i$, either $s_i = t_i$, and typing is direct, or $s_i
      \toq t_i$, in which case we can apply the induction hypothesis.
      One can thus replace each $p_{i_1\dots i_n}$ with its (potential) reduced
      $q_{i_1 \dots i_n}$; and as they are typed, and are reduction of
      orthogonal terms, thus still orthogonal, we are able to type back each
      (sup) rule, thus to type $t'$, and conclude. \qedhere
  \end{itemize}
\end{proof}

\begin{lemma}\label{lem:split-red}
  Let $t \equiv \sumdef[t]$. Suppose that $t \toqs v$.
  Then, there exist $s_i$ such that $t_i \toqs s_i$ and $v \equiv \sumdef[s]$.
\end{lemma}

\begin{proof}
  By induction on the length $k$ of the chain $t \toqs v$. Note that
  $\sumdef[s]$ is not a canonical form, as $s_i$ may contain superpositions.
  The case $k = 0$ being direct, suppose $t \toq t' \toq^k v$.
  As $t$ reduces, Lemma~\ref{lem:red-nequivz} implies that $\nequivz$, thus $t$
  has a canonical form, which can be built as in the proof of
  Lemma~\ref{lem:canonical}, i.e., discard any $t_i$ where $\inter{\alpha_i} =
  0$ or $\equivz[t_i]$, and for the remaining terms, write them as canonical
  forms up to complementation $t_i \equiv \sum_{j=1}^m \beta_{ij} \cdot p_j$.
  Therefore, $t \equiv \sum_{j=1}^m (\sum_{i=1}^n \alpha_i \beta_{ij}) \cdot
  p_j$.
  By Lemma~\ref{lem:red-can}, reduction of $t$ implies that $t' \equiv
  \sum_{j=1}^m (\sum_{i=1}^n \alpha_i \beta_{ij}) \cdot q_j$, with $p_j \toqif
  q_j$, again up to discarding possibly some $q_j$ where $\inter{\sum_{i=1}^n
  \alpha_i \beta_{ij}} = 0$.
  For each non-discarded $t_i$, remark that $t_i' = \sum_{j=1}^m \beta_{ij}
  \cdot q_j$ satisfies either $t_i'=  t_i$ if all $p_j \in \nf$, or $t_i \toq
  t_i'$.
  In any case, $t_i \toqs t_i'$ is satisfied, and up to readding the terms
  discarded at the beginning through $s \equiv t + 0 \cdot s'$, we obtain $t'
  \equiv \sum_{i=1}^n \alpha_i \cdot t_i'$, We can then conclude by applying
  the induction hypothesis on $t'$, and thus $t_i \toqs t_i' \toqs s_i$,
  therefore $t_i \toqs s_i$.
\end{proof}

\canonicaltyped*

\begin{proof}
  We prove that $\nequivz$ by induction on the typing rules of $t$.
  The fact that it implies that $t$ has a unique canonical form
  is a consequence of Lemma~\ref{lem:canonical}.
  Typing of each $t_i$ is obtained directly by induction on the typing of $t$,
  looking at the way the canonical form is built in Lemma~\ref{lem:canonical}.
  \begin{itemize}
    \item If $t = \mtx$, then we conclude directly as $\mta_1 \cdot \mtx$ is
      a canonical form.
    \item It $t \equiv s$, then $\nequivz[s]$ by induction
      hypothesis, and thus $\nequivz$ by quantity.
    \item Suppose $t = \mtb(t_1, \dots ,t_n)$. Looking at the proof of
      Lemma~\ref{lem:canonical}, $\equivz$ would imply that $\equivz[t_i]$ for
      some $i$, which is impossible by induction hypothesis.
    \item Finally, suppose $t = \sumdef$.
      By typing $\nequivz[t_i]$, thus we can take their canonical form, and up
      to adding $\mta_0 \cdot s$, we can write them as sharing the same
      classical terms: $t_i \equiv \sum_{j=1}^m \beta_{ij} \cdot s_j$, where
      $s_j$ are classical terms and pairwise distinct, and thus $t \equiv
      \sum_{j=1}^m \gamma_j \cdot s_j$, where $\gamma_j = \sum_{i=1}^n \alpha_i
      \beta_{ij}$.
      If $t \equiv 0 \cdot t$, this means that $\gamma_j = 0$ for all $j$.
      As $t_i \perp t_j$, given any substitution $\sigma \in \subcont[\Gamma
      \cup \Delta]$, $t_i \sigma \toqs v_i$.
      By Lemma~\ref{lem:split-red}, there exist $q_j$ such that $v_i \equiv
      \sum_{j=1}^m \beta_{ij} \cdot q_j$.
      Therefore, $v \triangleq \sum_{i=1}^n \alpha_i \cdot v_i \equiv
      \sum_{j=1}^m \gamma_j \cdot q_j \equiv 0 \cdot v$.
      However, expressing each $v_i$ in its canonical form, again by possibly
      complementing them, yields $v_i \equiv \sum_{k=1}^l \delta_{ik} \cdot
      r_k$, and thus $\equivz[v]$ implies that for all $k$, $\sum_{i=1}^n
      \alpha_i \delta_{ik} = 0$.
      Again, as $t_i \perp t_j$, $v_i \perp v_j$, and thus $\sum_{k=1}^l
      \delta_{ik} \delta_{jk} = 0$.
      Furthermore, by Definition~\ref{def:orthogonality}, $v_i$ has a canonical
      form, thus $\nequivz[v_i]$.
      Therefore, $(\gamma_{ik})_k$ are non-zeros orthogonal vectors.
      One can thus invert the matrix $G = (\gamma_{ik})_{ik}$, and the equality
      $\sum_{i=1}^n \alpha_i \delta_{ik} = 0$ would imply that all $\alpha_i =
      0$, which is not true as it is of norm $1$. Therefore,
      $\nequivz$. \qedhere
  \end{itemize}
\end{proof}

\progress*

\begin{proof}
  Prove the set equality by double inclusion, each time proving the
  contraposed result.
  Let us first show that any term that reduces is not equivalent to a value.
  This is done by induction on the reduction $\toq$.
  If $t = C[l \sigma]$, then $t$ is classical.
  In particular, if $t$ is equivalent to a value $v$, by quantity,
  $1 =\inter{\qt t t} = \inter{\qt t v} = 0$, as they are ground terms, thus
  $t \not\equiv v$.
  If $t \equiv t_1$, with $t_1$ reducing, by induction hypothesis, $t_1$ is not
  equivalent to a value, thus so is $t$, else $t_1$ would be.
  Finally, if $t = \candef$, reduction imposes that some $t_i$ is not
  equivalent to a value.
  If $t$ is equivalent to a value, by quantity,
  $0 \neq \inter{\alpha_i} = \inter{\qt{t_i} t} = \inter{\qt{t_i} v} = 0$,
  which concludes.

  Now, prove that any term not equivalent to a value does reduce.
  To do so, we first prove that any well-typed classical term $t$ that is
  not a value reduces through (C).
  By induction on the typing of $t = \mtb(t_1, \dots, t_n)$.
  Either $t_i$ is not a value, and thus $t_i = C[l \sigma]$, and it reduces to
  $C[r\sigma]$ for $l \to r \in \trsrule$. Taking $C' = \mtb(t_1,\dots, C,
  \dots, t_n)$ shows reduction of $t$.
  Else, if all $t_i$ are values, then $\mtb$ must be a function
  symbol, else $t$ is a value.
  As $\mfR$ is total, there exist $l \to r \in \trsrule$ and $\sigma$ such
  that $\mtf(t_1, \dots, t_n) = l \sigma$, thus $t$ reduces with $C = \diamond$.
  Now consider any general term $t$.
  As $t$ is well-typed, then it possesses a canonical form $\candef$ by
  Lemma~\ref{lem:canonical-typed}.
  Now, as $t$ is not equivalent to a value, at least one $t_i$ must be distinct
  from a value; as it is classical, the above result tells us that $t_i$
  reduces. Therefore, $\sumdef$ will reduce through (Q), as one element
  of the sum reduces; we conclude by reducing $t$ through (E).
\end{proof}

\subsection{Proofs of Section~\ref{sec:inference}}

\begin{lemma}\label{lem:type-procedure}
  Given a term $t$, there is a procedure, which may not terminate, to generate
  its typing derivation without any equivalence rule.
\end{lemma}

\begin{proof}
  By induction on the syntax of $t$.
  If $t = \mtx$, return a typing derivation $\typing * * t *$,
  indicating that $t$ can be typed with any context or type.
  If $t$ is the term we want to type, then it can be typed, by choosing any
  type. Suppose $t = \mtb(t_1, \dots, t_n)$.
  By induction hypothesis, either some $t_i$ cannot be typed, in which case
  so does $t$. Now, suppose each $t_i$ can be typed.
  By hypothesis, $t_i$ must be of type $T_i$, with $\signdef$, else it cannot
  be typed. In the case where $t_i = \mtx$, choose the according type, where
  its contexts containing a single variable, $\mtx$.
  Therefore, $\typing{\Gamma_i}{\Delta_i}{t_i}{T_i}$.
  One need to check that for any variable appearing in one of the contexts,
  either it appears in a single linear-context, or in multiple non-linear
  contexts, with the same type; if not, $t$ cannot be typed.
  This procedure is decidable, as the contexts are finite.
  Finally, denote $\Gamma = \cup_i \Gamma_i$; by weakening,
  $\typing \Gamma {\Delta_i}{t_i}{T_i}$ holds, and thus,
  $\typing \Gamma {\cup_i \Delta_i}{t}{T}$ holds.

  Finally, consider $t = \sumdef$.
  Again, if one $t_i$ cannot be typed, then so does $t$.
  Else, consider each typing derivation $\typing{\Gamma_i}{\Delta_i}{t_i}{T_i}$.
  If all $T_i$ are not identical or do not yield a quantum type, typing
  cannot be inferred.
  In the case where $t_i = \mtx$, either $n = 1$ and then type it as a qubit;
  if $n > 1$, $t_i$ cannot be all variables, else orthogonality will not hold,
  thus take the type from the other terms.
  Now, one need to check that the contexts can be reconciled as previously.
  Furthermore, we need to check that the phases and orthogonality conditions
  are verified between terms, which may be or not decidable.
\end{proof}

\inferenceundecide*

\begin{proof}
  Define the free typing inference as the problem of whether a term
  $t$ can be typed with a typing derivation without any equivalence
  rule. First, prove that free typing inference is $\pitwo$-complete.
  \cite{Sim09} showed that deciding whether a constructor TRS is
  strongly normalizing, i.e., terminates on all inputs, is
  $\pitwo$-complete.
  Suppose we have a STRS with one function symbol $\mtf$, and one wants to type
  $t = \stwo \cdot (\mtf(x), \zket) + \stwo \cdot (\mtf(x), \oket)$.
  If $t$ is typed, it must be typed via the superposition rule, thus
  $(\mtf(\mtx), \zket) \orthored (\mtf(\mtx), \oket)$, which
  imposes that $\mtf$ must terminate over any possible input,
  i.e., $\mfR$ is strongly normalizing.
  Therefore, if one can decide $\orthored$, thus typing, it can decide strongly
  normalization, thus free typing inference is $\pitwo$-hard.
  To prove completeness, we must show that it belongs to $\pitwo$.
  Taking the procedure from Lemma~\ref{lem:type-procedure},
  the only undecidable properties to check are verifying the conditions
  in the superposition rule.
  Deciding orthogonality between two terms $s, t$ is $\pitwo$, as it can
  be seen as follows:
  \[
    \forall \sigma, \exists k \in \mbN, s \sigma \toqs \candef[v] \wedge
    t \sigma \toqs \candefb w \wedge
    \sum_{i=1}^n \sum_{j=1}^m \inter{\alpha_i} \inter{\beta_j}^*
    \kronweird{v_i}{w_j} = 0
  \]
  The condition on the two summations is decidable, as we have restricted
  ourselves to terms in $\bar{\mbC}$, where summation, product and
  zero equality are decidable. Furthermore, $\kronweird - -$ is also
  decidable, as it can be checked inductively on the syntax of the term;
  therefore, it is $\pitwo$.
  The test for amplitudes condition is, by definition, and as we consider
  only algebraic numbers, $\Pi_1^0$.
  The overall check is thus $\pitwo$.
  As these undecidable checks are done a finite amount of times, as $t$
  has a finite syntax, the overall free type inference is $\pitwo$.

  Lemma~\ref{lem:tree-canonical} showed that any typing derivation can be
  written with an unique equivalence relation at the end, and with an
  above typing derivation which, in particular, has no equivalence relation.
  Therefore, type inference of $t$ implies that there exists $s$ such that $t
  \equiv s$ and $s$ is typable with no equivalence rule, thus it is
  $\Sigma_3^0$ by definition of the arithmetical hierarchy.
\end{proof}

\typeinference*

\begin{proof}
  Looking back at the procedure from Lemma~\ref{lem:type-procedure},
  one is able to ensure the following points:
  \begin{itemize}
    \item When one types $t$, the contexts we choose contain exactly
      the variables in $t$. Therefore, the size of the contexts are
      bounded by the size of $t$, and any check on whether contexts
      can be conciliated is done in $P(\size t)$.
    \item The check on amplitude symbols is done in polynomial time,
      as amplitudes are ground, thus there is a unique substitution
      (the empty substitution) for which the check is done.
      On algebraic numbers, this is computable, in a constant time;
      therefore, as the number of amplitudes to sum is bounded by
      the size of the term, it can be checked linearly.
    \item Computing the orthogonal Kronecker product on classical
      terms is done linearly in the syntax of the term.
  \end{itemize}
  Therefore, at each step of the induction, one does a finite number
  of recursive calls, and does a finite number of checks,
  each in polynomial time (denote the total complexity $Q$)
  Therefore, there exist a polynomial $P$ (depending on $Q$) such that
  type inference is done in $P(\size t)$.

  To show that $t$ is well-typed, we need to ensure that $\kronweird s t = 0$
  implies $s \orthored t$.
  Suppose $s = C[\zket, s_1, \dots, s_n]$ and $t = C[\oket, t_1, \dots, t_n]$
  without loss of generality, and take any substitution $\sigma$.
  As $\mfR$ is terminating, we obtain
  $s\sigma \toqs C[\zket, v_1, \dots, v_n]$ and
  $t\sigma \toqs C[\oket, w_1, \dots, w_n]$;
  By linearity, develop each $v_i$, and obtain the related canonical form,
  denoted $v, w$.
  Note that as $C$ is classical, the elements of their canonical forms are of
  the shape $v_i = C[\zket, v_1^i, \dots, v_n^i]$ (similar for $w$).
  One thus needs to compute their inner product.
  However, for any $v_i, w_j$, they are still of the shape $C[\zket, \dots]$
  and $C[\oket, \dots]$, thus $\kronweird{v_i}{w_j} = 0$, yielding directly
  $0$. Therefore, $s, t$ are orthogonal.
\end{proof}

\subsection{Proofs of Section~\ref{sec:circuit}}

\universality*

\begin{proof}
  By \cite[Section 5.3.1]{Ter03}, there exists a constructor TRS $\mfR_1$ of
  main symbol $\mtg$ such that for any $n \in \mbN$, $\mtg(\nats_n)$ outputs
  the representation of $C_n$, as $\toq$ behaves standardly for classical
  terms.
  As it is made only of classical rewrite rules, it can be seen and typed as
  a QTRS. Furthermore, as it contains finitely many constant symbols, they
  can be represented using a combination of lists and natural numbers, thus
  in $\mcC_Q$.
  As it treats only classical data, two distinct rules do not share the same
  input structure, and thus function symbols are structure preserving.
  By the same argument, all symbols quantum controls, taking $0$ control qubits.
  Let $\mcU$ be the universal set of gates used in the representation.
  Any circuit representation can be seen as a chain of wire swaps, applying a
  gate $U \in \mcU$, and another set of wire swaps.
  Any $n$-qubit unitary can be encoded as a QTRS with $2^n$ rules with no
  variables, and it will satisfy Definition~\ref{def:qtrs} as it is an unitary;
  we can encode a swap between the first two elements of a list easily; any
  general swap can then be built from this. All these function encode an
  unitary by definition, and are structure preserving.
  This allows us to build a QTRS $\mfR_2$ with main symbol $\mth$, where
  $\mth(C,t_{\ket \phi}) \toqs t_{C_n \ket \phi}$ for $C$ a
  circuit representation of $C_n$.
  As $\mth$ can be built with a unique rule $\mth(x) \to r$, with $r$ a
  composition of multiple gates, it satisfies the properties of $\qtrscirc$.
  We then conclude by taking the QTRS $\mfR$, where all the symbols and
  variables are joined, up to renaming, i.e. $\trsamp = \trsamp[\mfR_1]
  \biguplus \trsamp[\mfR_2]$, similarly for the other sets; and with one added
  function symbol $\mtf$, with one rewrite rule $\mtf(\nats_n, x) \to
  \mth(\mtg(\nats_n), x)$, which is the main symbol of $\mfR$.
\end{proof}

\compilebound*

\begin{proof}
  Let $w \in \structset \mtf$; denote $N = \size \trsrule + 1$ and $k
  = \size w$.
  We first produce inductively a circuit $C^t_w$ by induction on the syntax
  of $t$, and $C^\mtb_w$ for $\mtb \in \mcC \biguplus \mcF$;
  where the obtained circuit, for terms and function symbols, is of size
  bounded by $(\size{\trsrule}T(\size w))^{\rank[\mtg]} = K$, for
  $C^\mtg_w$ and $C^t_w$, where $\mtg$ is the function symbol in $t$ of maximum
  rank; the size of the $C^\mtc_w$ will be constant.
  From that, we approximate each gate, using the
  chosen universal set of gates, up to a precision $1 - \epsilon$, such that the
  overall circuit has a precision $\frac 2 3$. This approximation, done on each
  gate, can be done using Solovay-Kitaev theorem~\cite{Kit97}, which is known to
  approximate the original gate to a precision $1 - \epsilon$ in time
  $\mcO(\log^c(1/\epsilon))$, for a constant $c$. Choosing $\epsilon = \frac 2
  {3K}$ satisfies the overall precision~\cite{NC12}, and thus, our
  approximated circuit is of size $\mcO(K \log^c(\frac {3 K} 2)) =
  \mcO(T(k)^{\rank + 1})$.

  We may label the wires when they correspond to variables, to be able to
  merge circuits easily.
  Note that, at any step of the induction, the structure of any variable is
  known, as we know the starting structure, and $\fv l \supseteq \fv r$ for any
  rewrite rule.
  At each step, we may also do simplification steps, such as removing a gate
  controlled on a wire that is $\zket$. This removes any ill-path that would
  not terminate, or not in the correct number of steps.

  Start the induction with the symbol $\mtf$.
  As each step of induction corresponds to one reduction, and the QTRS
  terminates (this is a consequence of the structure preserving property), the
  induction will terminate.

  Suppose the induction step considers a classical term $t$.
  If $t = \mtx$, it yields the circuit with one wire, labelled with $\mtx$.
  If $t = \mtb(t_1, \dots, t_n)$, then inductively compute the circuit for each
  $t_i$, stack them vertically (in parallel), and merge their outputs with the
  inputs of the circuit generated by $\mtb$.
  As the structure of each $t_i$ is known, the structure of the inputs of
  $\mtb$ is also known, thus the number of input wires will be coherent with
  the total number of wires of the $t_i$, and thus wires can be merged.
  The number of rewrite steps satisfies $T(k) = \sum_{i=1}^n T_i + T'$,
  where $T_i$ is the number of rewrite steps of $t_i$, and $T'$ of $\mtb$.
  By induction, it thus satisfies
  $\size{C_w^t} = \sum_{i=1}^n \size{C_w^{t_i}} + \size{C_w^\mtb} \leq
  \sum_{i=1}^n(T_iN)^{\rank[\mtg]} + (T_iN)^{\rank[\mtg]} \leq ((N)\sum_{i=1}^n
  T_i + T')^{\rank[\mtg]} = (T(\size k)N)^{\rank[\mtg]}$,
  where $\mtg$ is the function symbol with the maximum rank between $\mtb$,
  and the function symbol of maximum rank in $t$.

  Let us now construct the circuits for symbols $\mtb \in \trscons \biguplus
  \trsfunc$. Suppose that $\mtb$ is a constructor, and enumerate all possible
  cases: $\natz, \nats, \nil$ yield no circuit nor wire; $\ket i$ initializes
  an ancillary qubit; for $(h,t)$ and $h :: t$, concatenate vertically both the
  wires of $h$ and of $t$; this is just wire concatenation, so no gate and thus
  of size $0$.

  Finally, consider the case where $\mtb$ is a function symbol $\mtg$.
  If $\mtg$-rewrite rules map to values, it can be implemented directly by an
  isometry, thus an unitary gate;
  remark that as we know the classical structure, any variable in an amplitude
  is fully known, and thus the amplitude can be interpreted as a scalar.
  Else, suppose that $\mtg$ quantum controls.
  Consider the $\mtg$-structural set of rewrite rules with the structure
  corresponding to the one fixed, and denote it
  $\mcR_\mtg = \{l_i \to r_i \in \mcR\}$.
  As it quantum controls, all of these rewrite rules only differ on their
  qubit constructors.
  Therefore, there exist contexts $C,C' \in \trsetclassical$ such that $l_i =
  C[\ket{j_{i1}, \dots, \ket{j_{im}}}]$ and $r_i = C'[\ket{j_{i1}}, \dots,
  \ket{j_{im}}, t_{i1}, \dots, t_{il}]$.
  One can thus view each $r_i$ as $r_i = \mth(\ket{j_{i1}}, \dots, \ket{j_{im}},
  t_{i1}, \dots, t_{il})$, with $\mth$ being a function symbol with one rule
  $\mth(x_1, \dots, x_m, q_1, \dots, q_l) \to C'[x_1, \dots, x_m, q_1, \dots,
  q_l]$, which corresponds to the computation.
  Therefore, for each $r_i$, compile each $t_{ik}$ to a circuit inductively,
  controlled by the $\ket{j_{i1}}, \dots, \ket{j_{im}}$,
  except for any possible sub symbol $\mtg$;
  as there is no symbol equivalent to $\mtg$ other than itself, $t_{ik}$,
  apart from $\mtg$, will be of rank $\rank[\mtg] - 1$; furthermore, $t_{ik}$
  terminates in $T_i$ steps, and $\sum_{k=1}^l T_{ik} \leq T(k)$.
  As this process is done for each $i$, thus at most $N - 1$ times, this yields
  a total circuit of size $(N-1)(\sum_{k=1}^l (T_{ik}N)^{\rank[\mtg] -1}) \leq
  (N-1)(T(k)N)^{\rank[\mtg] -1}$.
  Then, following the merging process from \cite[Theorem 9]{HPS25}, as all
  calls to $\mtg$ are done on an identical input, one is able to rewrite
  all calls to $\mtg$ as a single call; then, one compiles $\mtg$ inductively.
  This induction step can be done at most $T(k)$ times, as we have
  done one rewrite step; therefore, the size of the circuit is bounded by
  $\frac{(T(k) N)^{\rank[\mtg]}(N-1)}{N}$, which can be bounded by $(T(k)
  N)^{\rank[\mtg]} - 1$, as $\rank[\mtg] \geq 1$.
  After all the controlled statements of each $r_i$, we compile the circuit
  corresponding to $\mth$, corresponding to one gate, thus yielding the
  expected bound.
  Note that as $\mtg$ preserves the structure, each $t_{ik}$ yields the same
  number of wires.
  Furthermore, as $\mtg$ quantum controls, each left-handside contains the
  same variables, thus each $r_i$ has the same variables.
  To connect each circuit of $r_i$, one just has to connect the control wires
  together, the variable wires together, and finally to possibly merge the added
  ancillary qubit wires, up to adding NOT gates to flip them to
  $\zket$ or $\oket$. Finally, with $\mth$, merge the control wires together,
  and merge each wire from $t_{ik}$ with the wire labelled by $q_k$.
\end{proof}

\compile*

\begin{proof}
  The same compilation process can be done as in the proof of
  Theorem~\ref{thm:compile-bound}; as $\mfR$ terminates, it does in $\timeset$
  for some $\tau$; and $\qtrsrec$ only plays a role in the size of the circuit.
\end{proof}

\fbqpcharac*

\begin{proof}
  Let us first show soundness, i.e., let $\mfR \in \qtrspoly$,
  and show that any binary function computed by $\famcirc$ belongs to $\fbqp$,
  i.e., $\famcirc$ is an uniform polynomially-sized
  family of circuits. As $\mfR \in \qtrspoly$, there is a polynomial $Q$ such
  that $\mfR$ terminates in time $Q$. Therefore, by
  Theorem~\ref{thm:compile-bound}, the circuit is of size $\mcO(P(Q(n)))$,
  thus of polynomial size in $n$. Note that we have indexed the circuits by $n$
  rather than $w$, but $\size w = n$. The fact that this
  family is uniform is obtained by building the circuit as done in the proof of
  Theorem~\ref{thm:compile-bound}: one can build a Turing machine, taking as
  input the size $n$ (thus the shape), and suppose that it contains initially
  the encoding of $\mfR$, thus without the input, onto a tape of the QTM. As the
  circuit is build inductively on the syntax, and it differs only by the input
  structure, this creation is uniform, and is done in polynomial time, as the
  generated circuit is of polynomial size.

  \newcommand{\pket}{\ket \phi}
  \newcommand{\construct}[1]{\overline{#1}}
  \newcommand{\kbk}[3]{|#1\rangle\!\langle#2|#3\rangle}
  We now prove completeness, i.e., for any function $f$ in $\fbqp$, there exist
  $\mfR \in \qtrspoly$ such that $\famcirc$ computes $f$.
  We prove this result by using Yamakami's algebra~\cite{Yam20}. This paper
  defines a function algebra, and then proves that it completely characterizes
  \fbqp{}. Formally, it defines the function class $\square_1^{\mathrm{QP}}$,
  which is the smallest class of functions including the gates below, for
  $\theta \in [0, 2\pi) \cap \tilde{\mathbb C}$ where
  $\tilde{\mbC}$
  denotes the complex numbers whose real and imaginary parts can both be
  approximated by a polynomial-time Turing machine,
  \[
    \begin{aligned}
      I(\pket) &\triangleq \pket \\
      PHASE_\theta(\pket) &\triangleq \kbk 0 0 \phi + e^{i\theta}
      \kbk 1 1 \phi \\
      ROT_\theta(\pket) &\triangleq \cos \theta \ket \phi + \sin \theta (\kbk
      1 0 \phi - \kbk 0 1 \phi) \\
      NOT(\pket) &\triangleq \kbk 1 0 \phi + \kbk 0 1 \phi \\
      SWAP(\pket) &\triangleq
      \begin{cases*}
        \pket & if $l(\pket) \leq 1$\\
        \sum_{a, b \in \set{0,1}} \kbk{ab}{ba}{\phi} & otherwise
      \end{cases*}
    \end{aligned}
  \]
  and is closed under the following schemes:
  \[
    \begin{aligned}
      Compo[g, h](\pket) &\triangleq g \circ h(\pket) \\
      Branch[g, h] &\triangleq
      \begin{cases*}
        \pket & if $l(\pket) \leq 1$ \\
        \zket \otimes g(\braket{0|\phi}) + \oket \otimes
        h(\braket{1|\phi}) & otherwise
      \end{cases*} \\
      kQRec_t[g,h,p|\mcF_k](\pket) &\triangleq
      \begin{cases*}
        g(\pket) & if $l(\pket) \leq t$ \\
        h(\sum_{w \in \set{0,1}^k} \ket w \otimes
        f_w(\braket{w|p(\pket)})) & otherwise
      \end{cases*}
    \end{aligned}
  \]
  Here, $\pket$ is a quantum state, i.e., belongs to the Hilbert space $\mathbb
  C^{2^n}$ for a given $n$ written in Dirac notation. Its size $l(\pket)$ is
  $n$. Given $w \in \set{0,1}^n$ with $n \leq l(\pket)$, we may write
  $\pket = \sum_i \alpha_i \ket{w_i z_i}$, where $w_i \in \set{0,1}^n$ and
  $z_i \in \set{0,1}^{l(\pket) -n}$; we then abuse the notation by writing
  $\braket{w|\phi} = \sum_i \alpha_i \braket{w|w_i} \ket{z_i}$.

  This class is proven to be $\fbqp$ complete, thus any function of $\fbqp$ can
  be written as a function $\square_1^{\mathrm{QP}}$. We therefore associate
  any function $\square_1^{\mathrm{QP}}$ with a given QTRS,
  $\construct{\square_1^{\mathrm{QP}}}$, and prove inductively that they belong
  in $\qtrspoly$.
  By abuse of notation, we only write the corresponding set of rules of the
  QTRS below; these QTRS can be merged altogether, up to renaming.
  We also denote $\construct f_m$ as the main function symbol of the QTRS
  $\construct f$.

  One can remark that for any function $f \in \square_1^{\mathrm{QP}}$, and any
  input $w \in \set{0,1}^*$, $f\ket w = v$, and the main symbol of $\construct
  f$, $\construct f_m$, satisfies $\mtf(\computes w) \toqs v$, thus they
  compute the same function (here $\computes w$ is the notation from
  Theorem~\ref{thm:universality}).
  As $\construct{f}$ is approximated by the family of circuits $\mathtt
  C(\construct f)$ by Theorem~\ref{thm:compile-bound}, then $f$ is also
  approximated by this family, with precision $\frac 2 3$, thus giving us the
  wanted result.
  The following phases are used: $\inter{\mta_\theta} = e^{i \theta}$,
  $\inter{\mta^c_\theta} = \cos \theta$, $\inter{\mta^s_\theta} = \sin \theta$,
  $\inter{ \mathtt{sgn}} = -1$.
  \[
    \begin{aligned}
      \construct{I} &\triangleq \{\mtg(x) \to x \} \\
      \construct{Ph_\theta} &\triangleq \{\mtg(\nil) \to \nil,
        \mtg(\zket::t) \to \zket::t, \mtg(\oket::t) \to \mta_\theta
      \cdot \oket::t\} \\
      \construct{Rot_\theta} &\triangleq \left\{
        \begin{aligned}
          \mtg(\nil) &\to \nil \\
          \mtg(\zket::t) &\to \mta^c_\theta \cdot \zket :: t +
          \mta^s_\theta \cdot \oket :: t \\
          \mtg(\oket :: t) &\to \mta^s_\theta \cdot \mathtt{sgn}\cdot
          \zket :: t + \mta^c_\theta \cdot \oket :: t
        \end{aligned}
      \right\} \\
      \construct{Not} &\triangleq \{\mtg(\nil) \to \nil,
      \mtg(\zket::t) \to \oket::t, \mtg(\oket::t) \to \zket::t\} \\
      \construct{SWAP} &\triangleq \{\mtg(\nil) \to \nil,
      \mtg(h::\nil) \to h::\nil, \mtg(h::h'::t) \to h'::h::t\} \\
      \construct{COMP[f, g]} &\triangleq \{\mtg(x) \to
      \construct{f}_m(\construct{g}_m(x))\} \\
      \construct{Branch[f, g]} &\triangleq \left\{
        \begin{aligned}
          \mtg(\nil) &\to \nil \\
          \mtg(h::\nil) &\to h::\nil \\
          \mtg(\zket::h::t) &\to \zket :: \construct f_m(h::t) \\
          \mtg(\oket::h::t) &\to \oket :: \construct g_m(h::t)
        \end{aligned}
      \right\}
    \end{aligned}
  \]

  All terms belong in $\qtrscirc$ by construction, and as
  functions have no recursive construct, they terminate in a fixed time,
  independently of the input, and thus they belong in $\qtrspoly$.
  The definition of this function in \cite{Yam20} comes with $f_w \in \mcF_k$,
  where $\mcF_k$ is a set of functions $\{f_s\}_{s \in \set{0,1}^k}$, where
  each $f_s$ is either $I$ or $kQRec_t[g,h,p|\mcF_k]$ itself.
  Therefore, interpreting $f_s$ as either $\construct I$ or
  $\mathtt{rec}$, one is able
  to build the corresponding QTRS of main symbol $\mathtt{rec}$.

  \[
    \construct{kQRec_t[g,h,p|\mcF_k]} \triangleq \left\{
      \begin{aligned}
        \mtf(\nil) &\to \nil \\
        &\vdots \\
        \mtf(h_1 :: \dots h_k :: \nil) &\to h_1 :: \dots ::h_k :: \nil \\
        \mtf(s_1 :: \dots s_k :: l) &\to s_1 :: \dots :: s_k ::
        \construct{f_{s_1 \dots s_k}}_m(l) \\
        \mathtt{rec}(\nil) &\to \construct g_m(\nil) \\
        &\vdots \\
        \mathtt{rec}(h_1 :: \dots h_t :: \nil) &\to \construct
        g_m(h_1 :: \dots ::h_t :: \nil) \\
        \mathtt{rec}(h_1 :: \dots h_t :: l) &\to \construct
        h_m(\mtf(\construct p_m (h_1 :: \dots :: h_t :: l)))
      \end{aligned}
    \right\}
  \]
  Note that $h_i$ are variables, while $s_i$ are qubit constructors,
  thus $\mtf$ has $2^k + 2$ rules.
  By definition, this QTRS belongs to $\qtrscirc$, as the only structural
  set of rewrite rules with more than one rule is for the $2^k$ rules of $\mtf$
  with different $s_1, \dots, s_k$ which satisfy the definition;
  all other satisfy the definition of quantum controls with $C' = \diamond$,
  As recursive calls are done by $\mtf_q$ on a similar input, then it also
  belongs to $\qtrsrec$.
  Furthermore, terminating in polynomial time is guaranteed, as:
  \begin{itemize}
    \item By induction, each $g,h,p$ computes in polynomial time;
    \item Either $f_s$ is the identity, thus stops here, or is a recursive
      call on an input of size that decreases by $k$;
    \item Therefore, the compute time is roughly, for an input of size
      $n$, $\frac n k \max(P_g(n),P_h(p),P_p(n))$, where $P_g, P_h,P_p$ is the
      compute time of respectively $g,h,p$, which is polynomial, by
      induction hypothesis. The obtain time is thus polynomial in $n$. \qedhere
  \end{itemize}
\end{proof}

\section{Proofs of Section~\ref{sec:complexity}}

\chainsum*

\begin{proof}
  Suppose $t \toq t_1 \toq \dots$. By Lemma~\ref{lem:red-can}, it implies that
  $t \equiv \candefb s$, $t_1 \equiv \sumgen j m \beta {s'}$, and $s_j \toqif
  s_j'$.
  Now, if $t = \sumdef$, one can rewrite $t$ as its canonical form, i.e.,
  remove amplitudes equivalent to $0$, view $t_i \equiv \sum_{j=1}^m
  \gamma_{ij} \cdot s_j$ and $\beta_j = \sum_{i=1}^n \alpha_i \cdot
  \gamma_{ij}$.
  One can rewrite $t_1$ as $t_1 \equiv \sum_{i=1}^n \alpha_i \cdot
  (\sum_{j=1}^m \gamma_{ij} \cdot s_j')$; denoting $t_i' = \sum_{j=1}^m
  \gamma_{ij} \cdot s_j'$, one thus has $t_i \toqif t_i'$ (as we have removed
  $t_i$ where $\equivz[t_i]$, thus they have a canonical form and may reduce).
  Therefore, for each reduction of the infinite chain, each $t_i$ reduces via
  $\toqif$.
  We conclude as the reduction rule imposes that one $s_j$ actually reduces via
  $\toq$ and not $\toqif$, thus at least $t_i$ truly reduces at each step; as
  there is a finite number of $t_i$, one must reduce infinitely.
\end{proof}

\wpoprop*

\begin{proof}
  For the characterization of $\ordexteq$,
  the if condition can be shown by deriving $\ordexteq$ using all three rules
  from Definition~\ref{def:wpo} in the reverse order from the presentation;
  for the only if condition, it can be proven by induction hypothesis on the
  derivation of $\ordexteq$, by noting that any premise of any rule contains at
  least one classical term, which simplifies the characterization.
  All the other points can be seen as direct consequences of this
  characterization; for the last point, one can remark that $s \sigma \equiv
  \candef[s \sigma]$ as $s_i$ are classical terms and $\sigma$ maps to
  classical terms, and then conclude. \qedhere
\end{proof}

\wpotermination*

\begin{proof}
  Suppose there exists an infinite chain starting from $t_1$.
  By Lemma~\ref{lem:chain-sum}, writing its canonical form $\candef[s]$, there
  exists one $s_i$ that starts an infinite chain, thus $s \ordexteq s_i$.
  As it is classical, then $s_i = C[\sigma l]$, and let us denote $t_2 =
  C[\sigma r]$.
  As $\mfR$ is compatible with $\ordext$, $l \ordext r$; and as
  $\ord$ is a rewrite order, $\ordext$ is a rewrite order from
  Lemma~\ref{lem:wpo},
  thus $s_i \ordext t_2$.
  We therefore have an infinite chain $t_1 \ordexteq s_i \ordext t_2 \dots$.
  By transitivity, this yields a chain of classical terms
  $p_1 \ordext p_2 \dots$, which, by Lemma~\ref{lem:wpo}, implies that we have
  an infinite chain in $\ord$, which contradicts well-foundedness. \qedhere
\end{proof}

\begin{restatable}[]{lemma}{interwpo}\label{lem:inter-wpo}
  Let $\assign -$ be a monotonic assignment.
  Let $\ord$ be the restriction of $\ordassigneq$
  to classical terms.
  Then, $\ordexteq \subset \ordassigneq$,
  and $\ord$ is a well-founded rewrite order.
\end{restatable}

\begin{proof}
  Given $s \equiv \candef[s]$ and $t \equiv \candefb t$:
  \[
    \begin{aligned}
      \forall j, \exists i, s_i \ordassigneq t_i &\implies
      \forall j, \exists i, \assign{s_i} \ordnateq \assign{t_j} \\
      &\implies \max_{1 \leq i \leq n} \assign{s_i} \ordnateq \max_{1
      \leq j \leq m} \assign{t_j}
      \implies \assign s \ordnateq \assign t
      \implies s \ordassigneq t
    \end{aligned}
  \]
  As this is a characterization of a worst path ordering by Lemma~\ref{lem:wpo},
  $\ordexteq$ is thus encoded inside $\ordassigneq$.
  For any context $C$ and any term $s, t$ satisfying $s \ordassign t$:
  \[
    \assign{C[s]} = \assign{C}[\assign s] \ordnat \assign{C}[\assign t] =
    \assign{C[t]} \implies C[s] \ordassign C[t]
  \]
  where the first and last equality come from the definition of the assignment,
  and for the inequality, either $C$ is empty in which case it is trivial,
  either it is made of multiple symbols, and we can use monotonicity of
  $\assign -$.
  Given $s,t$ two terms and a substitution $\sigma$, $s \ordassign t$ implies
  that $\assign s \ordnat \assign t$, thus for any evaluation of the
  polynomials; and $\assign{s \sigma}$ can be seen as $\assign s$ where
  variables are evaluated correspondingly to $\sigma$.
  Finally, any infinite chain $t_1 \ordext t_2 \dots$ yields an infinite
  chain in $\ordnat$ (as all variables of $t_i$ are
    included in those of $t_1$, thus we can choose a unique set of substitutions
  for all rules), which is impossible. \qedhere
\end{proof}

\thminter*

\begin{proof}
  Direct by Lemma~\ref{lem:inter-wpo} and Theorem~\ref{thm:wpo}.
  While $\ordassigneq$ is not directly a worse path extension, it
  satisfies the induction hypothesis, thus the proof can be
  adapted to this specific case. \qedhere
\end{proof}

\rchain*

\begin{proof}
  The proof is similar to the one from \cite{AG00}.
  Let $t$ be a chain starting an infinite reduction.
  By Lemma~\ref{lem:chain-sum}, there exists $p \subcan t$ which also
  starts an infinite reduction.
  By a minimality argument, $p$ possesses a subterm $f_1(\seq{u_1})$,
  starting an infinite reduction, and where the sequence $\seq{u_1}$ terminates
  to values $\seq{z_1}$.
  Again by Lemma~\ref{lem:chain-sum}, there exists $\seq{v_1} \subcan
  \seq{z_1}$, such that $f_1(\seq{v_1})$ starts an infinite reduction (as
  $f_1(\seq{z_1})$ can be developed by linearity).
  Now, we have $f_1(\seq{w_1}) \to r_1 \in \mcR$, and a substitution
  $\sigma_1$, such that $f_1(\seq{w_1}) \sigma_1 = f_1(\seq{v_1})$, thus
  $f_1(\seq{v_1}) \toq r_1 \sigma_1$, and $r_1 \sigma_1$ starts an infinite
  reduction, and again, there is some classical term $y_1 \subcan
  r_1$ that starts
  an infinite reduction; as $\sigma_1$ maps to classical values, we
  also have $y_1
  \sigma_1 \subcan r_1 \sigma_1$.
  As $\sigma_1$ maps variables to values, and $y_1 \sigma_1$ starts an infinite
  reduction, there is a subterm, thus $C[f_2(\seq{u_2})] = y_1$ for some
  context $C$, such that $f_2(\seq{u_2})\sigma_1$ starts an infinite reduction,
  and $\seq{u_2} \sigma_1$ terminate.
  Therefore, let $\dpair{F_1(\seq{w_1})}{F_2(\seq{u_2})}$ be a valid dependency
  pair.
  We can repeat this process, as $f_2(\seq{u_2})$ starts an infinite reduction
  chain, thus obtaining an infinite sequence
  \[
    \dpair{F_1(\seq{w_1})}{F_2(\seq{u_2})},
    \dpair{F_2(\seq{w_2})}{F_3(\seq{u_3})}, \dots
  \]
  Finally, we need to prove that this is an $\mfR$-chain.
  By construction, $\seq{u_j}\sigma_{j-1} \toq^* \seq{z_j}$, with
  $\seq{w_j}\sigma_j \subcan \seq{z_j}$.
  One can assume, without loss of generality, that each rule is used with
  distinct variables, therefore we can take $\sigma = \sigma_1 \circ \sigma_2
  \circ \dots$, i.e., the disjoint union of all substitutions $\sigma_i$.
  Therefore, $F_j(\seq{u_j})\sigma \toq r$ with $F_j(\seq{w_j}) \sigma \subcan
  r$, thus this is an $\mfR$-chain. \qedhere
\end{proof}

\dpairorder*

\begin{proof}
  Suppose there is an infinite $\mfR$-chain $\dpairdef 1, \dpairdef 2, \dots$.
  Therefore, we have a substitution such that $t_j \sigma \toq^* r_j$ and
  $s_{j+1} \sigma \subcan r_j$ for all $j$.
  We will construct an infinite chain, using the worst path ordering
  generated from
  $\ordeq$.
  In particular, as $l \ordeq r_i$ for $l \to \candef[r] \in \mcR$, we have
  $l \ordexteq r$ for all rules.
  As $\ordeq$ is weakly monotonic and closed under substitutions, so is
  $\ordexteq$ by Lemma~\ref{lem:wpo}, and $t \toq^* s$ implies $t \ordexteq s$.
  Therefore, $t_j\sigma \ordexteq r_j \sigma \ordexteq s_{j+1}\sigma$, with the
  last inequality coming from the definition of a worst path ordering.
  Therefore, we can build the following chain:
  \[
    s_1 \sigma \ordext t_1 \sigma \ordexteq s_2 \sigma \ordext t_2 \sigma \dots
  \]
  thus, a chain $s_1 \sigma \ordext s_2 \sigma \dots$ between
  classical terms, thus
  an infinite chain in $\ord$, which contradicts well-foundedness. We then
  conclude by Theorem~\ref{thm:rchain}. \qedhere
\end{proof}

\redsum*

\begin{proof}
  This is proven by induction on the maximum number of reduction steps of
  the $t_i$, denoted $k$. If $k = 0$, they are all values, and we conclude.
  Suppose $k > 0$, thus one $t_i$ is not a value.
  For each $t_i$ such that $t_i \equiv \beta \cdot s$ with $\inter{\beta} = 0$,
  or where $\inter{\alpha_i} = 0$, we can remove it via $\equiv$,
  thus $t \equiv \sum_{k=1}^l \alpha_k \cdot t_k$.
  By writing each canonical form, up to completion (meaning we may add $\mta_0
    \cdot t'$ to the canonical forms so that they are defined on the same set of
  terms $s_j$), as $t_k \equiv \sum_{j=1}^m \beta_{k,j} \cdot s_j$,
  $s_j \toq s_j'$ or $s_j = s_j'$ is a value,
  and $\sum_{k=1}^l \alpha_k \cdot \beta_{k,j} = \gamma_j$:

  \[
    t \equiv \sum_{j=1}^m \gamma_j \cdot s_j \equiv \sum_{j, \gamma_j
    \neq 0} \gamma_j \cdot s_j
    \toq \sum_{j, \gamma_j \neq 0} \gamma_j \cdot s_j' \equiv
    \sum_{j=1}^m \gamma_j \cdot s_j' \equiv \sum_{k=1}^n \alpha_k \cdot t_k',
  \]

  as the $s_j$ are pairwise distinct, this is a canonical form.
  If the obtained form is a value,
  then we can stop here, and the result will be correct; even if $t_i$ is not
  a value, the non-value parts will still be canceled via $\equiv$ at the
  end, as the reduction is deterministic.
  Else, by $\equiv$, and we can replace $\mta_0 \cdot s_j$ by $\mta_0
  \cdot s_j'$.
  Furthermore, we can add back the removed terms from the beginning, adding
  back $t_i'$ instead of $t_i$, to get that
  $\sum_{k=1}^l \alpha_k \cdot t_k' \equiv \sum_{i=1}^n \alpha_i \cdot t_i'$.
  Thus, by $\equiv$, $t$ rewrites in one step to $\sumdef[t']$;
  and all $t_i'$ terminate in one less step, thus we conclude by
  induction hypothesis.
\end{proof}

\interbound*

\begin{proof}
  Let us prove by induction on $\assign t$, which is a natural number, that
  there is no $t$ such that $t$ terminates in at least $\assign t + 1$ steps;
  this implies the wanted result.
  Suppose $\assign t = 0$, thus $t$ reduces at least once.
  If $t$ is classical, then $t \toq s$, and by definition of the reduction rule
  (C), $\assign t > \assign s$; as $\assign s < 0$ and belongs to
  $\mbN$, this is not possible.
  If $t \equiv \candef$, at least one $t_i$ must reduce, else
  Lemma~\ref{lem:red-sum} would imply that $t$ terminates in $0$ steps; thus
  $t_i$ reduces, and $\assign{t_i} \leq \assign t$, thus we can conclude as
  above.
  Now, suppose that the result is proven for any $t'$ where
  $\assign{t'} \leq k$,
  and take $t$ where $\assign t = k + 1$.
  If $t$ is classical, then, as before, there exists $s$ such that $t \toq s$,
  thus $\assign s < \assign t$, and $s$ terminates in at least $\assign t \geq
  \assign s + 1$ steps, which by induction hypothesis cannot happen.
  If $t$ is non-classical, proceed as previous, by picking one $t_i$ from
  the canonical form that reduces in at least $\assign t + 1 \geq
  \assign{t_i} + 1$
  steps, which must exist by Lemma~\ref{lem:red-sum}.
\end{proof}
\begin{lemma}\label{lem:additive}
  Let $\mfR$ be a QTRS, and let $\assign -$ be an assignment.
  Then, there exists $k \in \mbN$, such that for any classical value
  ${v}$, $\assign{{v}} < k \times \size{{v}}$.
\end{lemma}

\begin{proof}
  By induction on the syntax of ${v} $; take $k = 1 + \max_{\mtc \in \mcC}
  \alpha_\mtc$.
\end{proof}

\polytime*

\begin{proof}
  By Lemma~\ref{lem:inter-bound}, $f({v}_1,\ldots,{v}_n)$ terminates
  in at most $\assign{f({v}_1,\ldots,{v}_n)}$ steps.
  It satisfies:
  \[
    \assign{f({v}_1,\ldots,{v}_n)} =
    \assign \mtf(\assign{{v}_1}, \dots, \assign{{v}_n})
    < \assign \mtf(k\size{{v}_1}, \dots, k\size{{v}_n})
    = P(\size{{v}_1}, \dots, \size{{v}_n}),
  \]
  where the second inequality comes from Lemma~\ref{lem:additive} and
  monotonicity of $\assign -$, and the last inequality as $\assign \mtf$ is a
  polynomial $Q$, and write $P$ as the polynomial which replaces each
  indeterminate $X$ by $kX$, which is also a polynomial. \qedhere
\end{proof}

\end{document}